\documentclass[11pt]{article}
\usepackage{amsfonts,amsmath}
\usepackage{hyperref}
\usepackage{latexsym}
\usepackage{amsthm}
\usepackage{amscd}
\usepackage{epsfig}
\usepackage{tikz}
\usepackage{tikz-cd}
\usepackage{mathtools}
\usepackage{graphicx}
\usepackage{caption}
\usepackage{enumerate}
\usepackage[shortlabels]{enumitem}
\usepackage{subcaption}
\usetikzlibrary{decorations.pathreplacing}
\usepackage{comment}

\usepackage{float}
\usepackage{soul}
\usepackage{bm}
\usepackage{planar_integrability}
\usepackage{dsfont}

\numberwithin{equation}{section}

\newtheorem{thm}{Theorem}[section]
\newtheorem{prop}[thm]{Proposition}
\newtheorem{lem}[thm]{Lemma}
\newtheorem{cor}[thm]{Corollary}
\newtheorem{defn}[thm]{Definition}
\newtheorem{conj}[thm]{Conjecture}

\newcommand{\be}{\begin{equation}}
\newcommand{\ee}{\end{equation}}
\newcommand{\bib}{\bibitem}
\newcommand{\nc}{\newcommand}

\nc{\Cbb}{\mathbb{C}}
\nc{\Fbb}{\mathbb{F}}
\nc{\Nbb}{\mathbb{N}}
\nc{\Rbb}{\mathbb{R}}
\nc{\Zbb}{\mathbb{Z}}
\nc{\Ac}{\mathcal{A}}
\nc{\Bc}{\mathcal{B}}
\nc{\Cc}{\mathcal{C}}
\nc{\Dc}{\mathcal{D}}
\nc{\Fc}{\mathcal{F}}
\nc{\Hc}{\mathcal{H}}
\nc{\Jc}{\mathcal{J}}
\nc{\Lc}{\mathcal{L}}
\nc{\Mc}{\mathcal{M}}
\nc{\Oc}{\mathcal{O}}
\nc{\Rc}{\mathcal{R}}
\nc{\Sc}{\mathcal{S}}
\nc{\Vc}{\mathcal{V}}
\nc{\Zc}{\mathcal{Z}}
\nc{\Pb}{\mathsf{P}}
\nc{\tr}{\mathrm{tr}}
\nc{\gh}{\hat{g}}
\nc{\id}{\mathrm{id}}
\nc{\idn}{\mathds{1}_n}
\nc{\wj}{\mathsf{w}}
\nc{\Fo}{\overline{F}}
\nc{\ko}{\overline{k}}
\nc{\Ko}{\overline{K}}
\nc{\wo}{\overline{w}}
\nc{\Wo}{\overline{W}}
\nc{\xo}{\overline{x}}
\nc{\Xo}{\overline{X}}
\nc{\yo}{\overline{y}}
\nc{\Yo}{\overline{Y}}
\nc{\at}{\tilde{a}}
\nc{\Tt}{\tilde{T}}
\nc{\io}{\overline{i}}
\nc{\ioa}{\overline{1}}
\nc{\iob}{\overline{2}}
\nc{\ioc}{\overline{3}}
\nc{\iod}{\overline{4}}
\nc{\ab}{\mathbf{a}}
\nc{\mc}{\check{m}}
\nc{\cf}{\mathsf{c}}
\nc{\pf}{\mathsf{p}}
\newcommand{\TLmodule}{V}

\newcommand{\TLgram}{G}

\newcommand{\TLposet}{D}

\newcommand{\TLset}{L}

\definecolor{dblue}{rgb}{.61,.61,1}
\definecolor{lightblue}{rgb}{.61,.61,1}
\definecolor{altblue}{rgb}{.61,.61,1}

\begin{document}

\thispagestyle{empty}

\begin{center}

\textbf{\huge Integrability of planar-algebraic models}
\\[0.6cm]
{\LARGE Xavier Poncini and J{\o}rgen Rasmussen}
\\[0.3cm]
{\it School of Mathematics and Physics, University of Queensland\\ St Lucia, Brisbane, Queensland 4072, Australia}
\\[0.3cm] 
{\sf x.poncini\!\;@\!\;uq.edu.au\quad\ j.rasmussen\!\;@\!\;uq.edu.au}

\vspace{1.4cm}

{\large\textbf{Abstract}}\end{center}
The Quantum Inverse Scattering Method is a scheme for 
solving integrable models in $1+1$ dimensions, building on an $R$-matrix that satisfies the Yang--Baxter equation
and in terms of which one constructs a commuting family of transfer matrices.
In the standard formulation, this $R$-matrix acts on a tensor product of vector spaces.
Here, we relax this tensorial property and develop a framework for describing and analysing 
integrable models based on planar algebras, allowing non-separable \textit{$R$-operators} 
satisfying \textit{generalised} Yang--Baxter equations.
We also re-evaluate the notion of integrals of motion and characterise when an (algebraic) \textit{transfer operator}
is polynomial in a single integral of motion. We refer to such models as {\em polynomially integrable}.
In an eight-vertex model, we demonstrate that the corresponding transfer operator is polynomial
in the natural hamiltonian.
In the Temperley--Lieb loop model with loop fugacity $\beta\in\Cbb$, we likewise find
that, for all but finitely many $\beta$-values, the transfer operator is polynomial in the usual hamiltonian element 
of the Temperley--Lieb algebra $\mathrm{TL}_n(\beta)$, at least for $n\leq17$.
Moreover, we find that this model admits a second canonical hamiltonian, and that this hamiltonian also
acts as a polynomial integrability generator for small $n$ and all but finitely many $\beta$-values.

\newpage

\tableofcontents

\newpage

\section{Introduction}
\label{Sec:Intro}

Based on Bethe's pioneering work \cite{Bethe31}, the Quantum Inverse Scattering Method 
\cite{FST79,FT79,Sklyanin82,KBI93} is a collection of techniques for solving integrable models in $1+1$ dimensions.
It is built on an $R$-matrix that acts on a tensor product of vector spaces,
allowing the construction of a commuting family of transfer matrices whose spectral properties can 
be used to gain insight into the physical properties of the model.
Although hugely successful, this approach does have its limitations.
Using ideas of Jones, this was illustrated in the work \cite{PRZ06} on a Temperley--Lieb loop model, 
where the basic building block replacing the $R$-matrix is a parameterised element of a 
planar algebra \cite{Jones99} whose representations do not naturally separate tensorially.
This avenue was explored further in a series of papers by Pearce et al.,
including \cite{PR07,PRV10,PRT14,MDPR14,PRT14b}, where commutativity of the corresponding transfer 
operators is established in the algebraic setting and is therefore valid in any given representation. 
In the present work, we generalise and formalise this approach by developing a framework for describing and analysing 
two-dimensional integrable models whose algebraic structure is afforded by a planar algebra.
In another recent study of integrable models, the underlying mathematical structure is a \textit{fusion category} \cite{AFM20,Fendley21}.

There is no consensus as to what constitutes quantum integrability, see \cite{CM11} for a 
review. Here, we approach the problem from a statistical mechanical set-up based on a two-dimensional model
described by an algebra $\Ac$. 
We will view such a model as defined by a one-parameter family of transfer operators, $T(u)\in\Ac$, and say that
the model is {\em integrable} if $[T(u), T(v)] = 0$ for all $u$ and $v$ in some suitable domain. 
From $T(u)$, one may extract a $u$-independent element of $\Ac$ and interpret it as
the \textit{quantum hamiltonian} of a one-dimensional system.
To contrast, another common, and in some sense complementary, approach is to start with a hamiltonian of a 
one-dimensional quantum system, such as a quantum spin chain, 
and then construct a transfer operator as a generating function for the \textit{integrals of motion}.
For precision in our setting, we re-evaluate the notion of integrals of motion algebraically and distinguish between 
those arising from the model-defining transfer operator, which we generally refer to as the hamiltonians of the 
model, and those in the centraliser of the hamiltonians within the algebra $\Ac$.

In a \textit{Yang--Baxter integrable} model, the commutativity of the transfer matrices is a consequence 
of a set of local relations satisfied by the $R$-matrices, including the celebrated Yang--Baxter equations 
(YBEs) \cite{Mcguire64,Yang67,Baxter71,Baxter07}.
Depending on the model, these may be supplemented by Boundary YBEs \cite{Sklyanin88} 
involving $K$-matrices encoding boundary conditions of the model.
This extends to similar relations between 
the $R$- and $K$-{\em operators} in the Temperley--Lieb algebraic approach in \cite{PRZ06}.
Within the planar-algebraic framework developed here, we work with \textit{generalised} 
YBEs which allow the auxiliary operators (the `middle' operators) to be parameterised differently from
the $R$-operators.

A key aim of this work is to examine the possibility of the transfer operator being polynomial in 
a spectral-parameter independent integral of motion,
and we say that a model with that property is {\em polynomially integrable}.
From this view, the ensuing integrability (in the sense of commuting transfer operators) is a consequence
of a {\em global} property, in contrast to Yang--Baxter integrability which is {\em local} in origin. 
Despite the differing perspectives, we stress that they are often applicable simultaneously, 
as a model can be {\em both} Yang--Baxter {\em and} polynomially integrable. 
In fact, we find (under mild conditions on the underlying algebra $\mathcal{A}$) 
that polynomial integrability follows from the diagonalisability of an integrable transfer operator
-- a ubiquitous property in a large class of (possibly Yang--Baxter) integrable models. 
The indicated shift in perspective will be explored in two concrete examples: 
a Temperley--Lieb loop model (Section \ref{Sec:TL}) and an eight-vertex model (Appendix \ref{Sec:Tensor}).

The Temperley--Lieb loop model is formulated in the 
Temperley--Lieb planar algebra \cite{Jones99} incorporating the usual Temperley--Lieb algebra 
$\mathrm{TL}_n(\beta)$ \cite{TL71,Jones83} and underlying the body of work initiated in \cite{PRZ06}.
We show that the model is 
{\em freely Baxterisable} in the sense that it is integrable for any choice of $R$-operator parameterisation.
We then show that the model is also polynomially integrable for all $\beta>2\cos\frac{\pi}{n}$ and $u\in\Rbb$.
Using a new decomposition of the transfer operator $T_n(u,\beta)$, we classify the identity points (see below)
for $T_n(u,\beta)$ and find that there are
{\em two} canonical hamiltonians, one of which is the familiar one, $h=-(e_1+\cdots+e_{n-1})$, 
while the other one does not appear to have been discussed in the literature.
Using the cellularity \cite{GL06} of $\mathrm{TL}_n(\beta)$, we present a spectral analysis of the hamiltonians,
with emphasis on their minimal polynomials. Armed with this, we show that 
$T_n(u,\beta)$ is polynomial in $h$ for all but finitely many $\beta\in\Cbb$, at least for $n\leq17$.
We also examine the finitely many $\beta$-values for which $T_n(u,\beta)$ fails to be polynomial in $h$ and find 
that $T_n(u,\beta)$ is polynomial in the alternative hamiltonian for those values, at least for small $n$.

The paper is organised as follows.
In Section \ref{Sec:Planar}, we recall basic properties of planar algebras
(with some details deferred to Appendix \ref{Sec:PAs})
and develop our planar-algebraic 
description of Yang--Baxter integrability. We thus introduce suitable $R$- and $K$-operators and
use these to construct the transfer operators. A finite set of local relations, including generalised YBEs, 
provide sufficient conditions for integrability (Proposition \ref{Prop:YBE}). 
An {\em identity point} of a transfer operator is introduced as a value for the spectral parameter at which the 
transfer operator evaluates to a scalar multiple of the identity operator, and we use power-series expansions about 
these points to determine the corresponding {\em principal hamiltonians}.
An algebraic characterisation of integrals of motion and the introduction
of {\em polynomial integrability} conclude the section.

In Section \ref{Sec:Linear}, we consider polynomial integrability in a linear-algebraic setting and
present necessary and sufficient conditions for parameter-dependent matrices to be expressible 
as polynomials in a parameter-{\em independent} matrix (Proposition \ref{prop:BxMn}). We then extend this result to elements of semisimple 
algebras (Proposition \ref{prop:Asemi}). We also prepare for our study of the Temperley--Lieb loop model and the aforementioned eight-vertex 
model, by introducing notions relevant to the employed spectral analysis and by reviewing basic properties of 
cellular algebras.

In Section \ref{Sec:TL}, we turn to the Temperley--Lieb algebra where we find that the associated loop model is 
both freely Baxterisable (Proposition \ref{Prop:TT0}) and, for $\beta$ generic, polynomially integrable (Proposition \ref{prop:TLIntSuff}). 
The identity points of the model are then classified (Proposition \ref{prop:TLIdPts}), 
and two distinct hamiltonians are identified as candidates for an integral of motion in terms of which the 
transfer operator can be expressed as a polynomial. Spectral analysis confirms that both
objects do the job for all but finitely many $\beta\in\Cbb$ and all $n\le 17$. 
We then determine explicit polynomial expressions of the transfer operator in terms of these hamiltonians
and find that, for all $\beta\in\Cbb$ and $n$ small, at least one of the polynomials is well-defined.
Some technical details are deferred to Appendix \ref{Sec:App}.

Section \ref{Sec:Discussion} contains concluding remarks,
while in Appendix \ref{Sec:Tensor}, we recover the familiar quantum inverse scattering framework by specialising 
to {\em tensor planar algebras}. Our primary example of this type is an eight-vertex model, 
which we use to illustrate 
how the planar-algebraic framework simplifies in case the $R$-operators are tensorially separable.
We show that the transfer matrix is diagonablisable and present an exact solution in the form of explicit 
expressions for all eigenvalues and corresponding eigenvectors of the transfer matrix.
We then exploit similarities in the spectral properties of the transfer matrix 
and the canonical hamiltonian to establish that the transfer matrix is polynomial in the hamiltonian.
Moreover, we decompose the transfer matrix into a linear combination
of a complete set of orthogonal idempotents expressed in terms of the minimal polynomial of the hamiltonian.

Throughout, we let $\Nbb$ denote the set of positive integers, $\Nbb_0$ the set of nonnegative integers, 
$\Fbb$ an algebraically closed field of characteristic zero, and $\Rc$ a commutative ring with identity.

\section{Planar algebras}
\label{Sec:Planar}

Here, we develop a general planar-algebraic description of Yang--Baxter integrability.
Much of the material would be known to experts, but the details and generality of the framework does
not seem to have been outlined before in the literature. 
This includes the formulation and treatment of the transfer operators in Section \ref{Sec:Transfer}, 
the {\em generalised} Yang--Baxter equations in Section \ref{Sec:Baxter},
the definition of hamiltonians and integrals of motion in Sections \ref{Sec:HamiltonianSub} and \ref{Sec:Higher},
and the introduction of {\em polynomial integrability} in Section \ref{Sec:Pol}.
To set the stage and fix our notation, we first recall some basic properties of
planar algebras \cite{Jones99,Jones00,JonesNotes}.
In the following, focus will be on so-called {\em unshaded} planar algebras, 
while some technical details are deferred to Appendix \ref{Sec:PAs}.

\subsection{Algebraic structure}
\label{Sec:PlanarBasics}

Informally, a \textit{planar algebra} \cite{Jones99,Jones00,JonesNotes} is a collection of vector spaces 
$(P_n)_{n\in\Nbb_0}$ where vectors can be `multiplied planarly' to form vectors.
A basis for $P_n$ consists of disks with $n$ \textit{nodes} (connection points) on their boundary, 
whereby a boundary is composed of nodes and boundary intervals. 
The specification of the internal structure of the basis disks is a key part of the definition of any given planar 
algebra. When disks are combined (`multiplied'), 
every node is connected to a single other node (possibly on the same disk)
via non-intersecting strings, and \textit{planar tangles} are the diagrammatic objects, defined up to ambient isotopy, 
that facilitate such combinations. 

Planar tangles have the following general features: 
(i) an exterior disk, called the \textit{output} disk, 
(ii) a finite set of interior disks, called \textit{input} disks, 
(iii) a finite number of non-intersecting strings connecting nodes of disks pair-wise or forming loops not touching any of the disks, and 
(iv) a marked boundary interval on each disk. An example of a planar tangle is given by
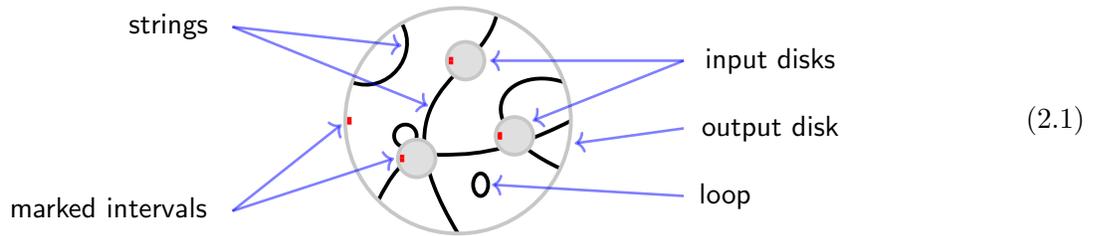
\begin{align}\label{equ:TtangleExample}
    \raisebox{-1.4275cm}{\PlanarTangleExample}
\end{align}
We denote the output disk of the planar tangle $T$ by $D_0^T$ and the set of input disks by $\Dc_T$. 
The number of nodes on the (exterior) boundary of a disk $D$ is denoted by $\eta(D)$. 
The boundary marks on each disk disambiguate the alignment of the input and output disks and are
indicated by red rectangles, see (\ref{equ:TtangleExample}). 

Planar `multiplication' is induced by the action of the planar tangles as multilinear maps.
For a planar tangle $T$, this is denoted by
\be
 \Pb_T\colon\bigtimes_{D\in\Dc_T}P_{\eta(D)}\to P_{\eta(D_0^T)}.
\label{PT}
\ee
Pictorially, $\Pb_T$ acts by filling in each of the interior disks $D\in\Dc_T$ with an element of the corresponding 
vector space $P_{\eta(D)}$, in such a way that the nodes match up and the marked intervals are aligned. 
The details of the map $\Pb_T$ specify how one should remove the internal disks in the picture and identify the 
image as a vector in $P_{\eta(D_0^T)}$. If $\Dc_T=\emptyset$, we simply write $\Pb_T()$ for the image under 
$\Pb_T$, and stress that $\Pb_T()$ is distinct from $T$.
Taking $T$ as in \eqref{equ:TtangleExample}, we present the example
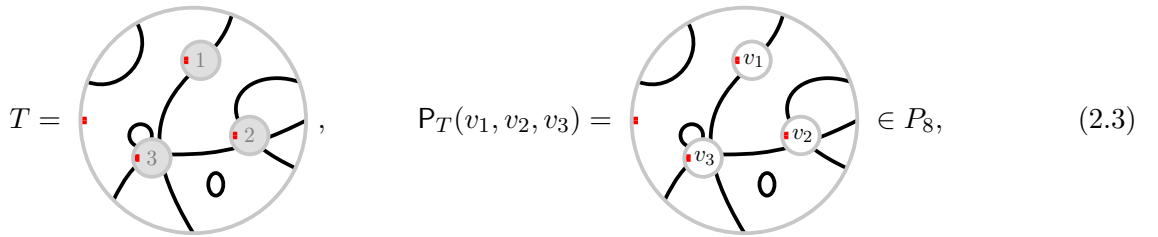
\begin{align}
    T= \raisebox{-1.4275cm}{\RotPlanarNKTangleExampleAltCNew}, \qquad\quad 
    \Pb_T(v_1,v_2,v_3)
    =  \raisebox{-1.4275cm}{\RotPlanarNKTangleExampleAltCGenNew} \in P_8,
    \label{equ:GenPlanarActionExample}
\end{align}
where $v_1\in P_2$, $v_2\in P_4$ and $v_3\in P_6$. 
Note that we have not specified any details about $(P_n)_{n\in\Nbb_0}$ or about the action of the planar tangles. 
\medskip

\noindent
\textbf{Remark.}
Unlike in the picture of $T$ in (\ref{equ:GenPlanarActionExample}), disks in $\Dc_T$ are not labelled; however,
to apply the ordered-list notation for the vectors in $\Pb_T(v_1,v_2,v_3)$, it is convenient to label the disks
accordingly. Once drawn as in the second picture in (\ref{equ:GenPlanarActionExample}), no labelling is needed.
\medskip

Planar tangles can be `composed', and consistency between this composition and the associated 
multilinear maps is often referred to as {\em naturality}, see Appendix \ref{Sec:Naturality} for details. 
By specifying the vector spaces $(P_n)_{n\in\Nbb_0}$ and the action of planar tangles as 
multilinear maps \eqref{PT} such that naturality is satisfied, one arrives at a \textit{planar algebra}. 
As an example, for each $n\in\Nbb_0$, let $P_{2n}$ denote the vector space spanned by all planar tangles $T$ 
with no input disks and $\eta(D_0^T)=2n$, while $P_{2n+1}=\{0\}$, and let the multilinear maps act by the 
composition of tangles. By construction, this data satisfies naturality.
Revisiting \eqref{equ:GenPlanarActionExample} for this particular planar algebra, we have
\begin{align}\label{equ:PlanarActionExample}
    \Pb_T\big(\!\!\raisebox{-0.2cm}{\RotPlanarExampleNodeIIIAltCNew}\!,\!
    \raisebox{-0.2cm}{\RotPlanarExampleNodeIAltC}\!,\!
    \raisebox{-0.2cm}{\RotPlanarExampleNodeIIAltC}\!\!\big) 
    = \raisebox{-1.4275cm}{\RotPlanarNKOperandExampleAltCNew} 
    =  \raisebox{-1.4275cm}{\RotPlanarNKOperandExampleAltCnewNew} 
    = \raisebox{-0.2cm}{\RotPlanarExampleNodeIVAltC}\!.
\end{align}
In this example, each vector space $P_{2n}$ is infinite-dimensional, as a set of disks 
with identical internal connections but different numbers of loops is considered linearly independent. 

A natural condition on a planar algebra is $\dim P_0=1$.
In this case, the empty disk and that with any number of loops are linearly dependent.
Accordingly, there exists a linear map $\mathrm{e}\colon P_0\to\Fbb$, here referred to as the 
{\em evaluation map}, that maps the `empty tangle' to $1$, and we may identify $P_0$ with $\Fbb$. 
Naturality implies that each loop formed in the image under $\mathsf{P}_T$ can be removed and replaced by a 
common scalar weight. Another natural condition on a planar algebra is to forbid \textit{null vectors}.
Here, a nonzero $v\in P_n$ is called a \textit{null vector} if $\mathsf{P}_T(v)=0$ for every planar tangle $T$
for which $\Dc_T=\{D\}$ with $\eta(D) = n$.

In the remainder of this work, all planar algebras will have $\dim P_0=1$ and $P_{2n+1}=\{0\}$,
and $P_{2n}$ will have no nonzero null vectors for all $n\in\Nbb_0$. To distinguish such a planar algebra from the 
general discussion above, we will use the notation $(A_n)_{n\in\Nbb_0}$, where $A_n\equiv P_{2n}$ is a vector 
space spanned by disks with $2n$ marked boundary points.

A planar algebra $(A_n)_{n\in\Nbb_0}$ is not an algebra in the usual sense; however, it contains a countable
number of standard algebras. Indeed, for each $n\in\Nbb_0$, the planar tangle
\begin{align}\label{equ:MultPlanar}
    M_n: = \raisebox{-1.4275cm}{\RotPlanarIIKMultAltCNew}, \qquad\quad 
    \Pb_{M_n}\colon A_n\times A_n\to A_n,
\end{align}
induces a \textit{multiplication} on $A_n$, and we write $vw=\Pb_{M_n}(v,w)\in A_{n}$,  
where $v$, respectively $w$, is replacing the lower, respectively upper, disk in $M_n$. 
Naturality implies that the ensuing standard algebra $A_n$ is associative and unital, where
\be
    \idn:= \Pb_{\mathrm{Id}_n}\!(), \qquad\quad 
    \mathrm{Id}_n := \raisebox{-0.675cm}{\IdZTangleAltIINew}.
\label{Aunit}
\ee
Details are deferred to Appendix \ref{Sec:Unitality}.

\subsection{Transfer operators}
\label{Sec:Transfer}

Let $(A_n)_{n\in\Nbb_0}$ be a planar algebra, and, for each $n$, let $B_n$ denote a basis for $A_n$. 
Without loss of generality, we may assume that $\idn\in B_n$, and it is then convenient to introduce
\be
 B_n':=B_n\setminus\{\idn\}.
\label{Bnprime}
\ee
We also introduce the parameterised algebra elements
\be
 K(u): = \sum_{a\in B_1} k_a(u)\,a,\qquad
 R(u): = \sum_{a\in B_2} r_a(u)\,a,\qquad
 \Ko(u): = \sum_{a\in B_1} \ko_a(u)\,a,
\label{RKK}
\ee
where $k_a, r_a,\ko_a:\Omega\to\Fbb$, with $\Omega\subseteq\Fbb$ a suitable domain,
and the {\em transfer tangle}
\be
 T_n := \raisebox{-1.975cm}{\PlanarTemplateTransferAlt}.
\label{Tn}
\ee
We are now in a position to define the corresponding {\em transfer operator} as
\be
 T_n(u): =\Pb_{T_n}\big(K(u),R(u),\ldots,R(u),\Ko(u)\big),
\label{Tnu}
\ee
where $K(u)$ and $\Ko(u)$, respectively, are inserted into the left- and rightmost disk spaces in (\ref{Tn}).
It follows that $T_n(u)$ is an element of $A_n$.
Using the diagrammatic representations
\begin{align}
 K(u)=\raisebox{-0.425cm}{\RotKopBlueOp}\; ,\qquad
 R(u)=\raisebox{-0.325cm}{\RotRopGreenOpBBTLST}\; ,\qquad
 \Ko(u)=\raisebox{-0.425cm}{\RotKopYellowOp}\; ,
\label{RKK2}
\end{align}
we have
\begin{align}
 T_{n}(u)\,= \raisebox{-0.925cm}{\RotMarkRedTransferFilteredCrossIsoSingleR}
   = \raisebox{-2.55cm}{\RotMarkRedTransferFilteredCrossSingleR}\ ,
\label{TnuJaws}
\end{align}
and we refer to $R(u)$ as the corresponding \textit{$R$-operator}, and to $K(u)$ and $\Ko(u)$ as the corresponding \textit{$K$-operators}.

\medskip

\noindent
\textbf{Remark.} 
Different internal colours in (\ref{RKK2}) are used to indicate that the parameterisations are 
independent. In fact, it is for later convenience that we use the {\em same} $A_2$-element $R(u)$ in every 
available position in (\ref{Tnu}) and hence in (\ref{TnuJaws}),
and that we let the coefficients in (\ref{RKK}) be functions of a {\em single} variable only.
Generalisations are possible and in some cases straightforward but will not be considered here.
\medskip

Within this planar-algebraic framework, the transfer operator can be recast as a partial-trace expression
reminiscent of more traditional descriptions. 
To this end, we introduce the \textit{trace} and \textit{partial trace tangles}
\begin{align}\label{equ:TracePlanar}
        \tr_n: =\raisebox{-1.2cm}{\RotPlanarITraceAltC}, \qquad\quad 
        \tau_n: =\raisebox{-1.2cm}{\RotPlanarIPTraceAltC},
\end{align}
and the {\em embedding tangles}
\be
 E_{n,j}^{(1)}: =\raisebox{-1.95cm}{\PlanarTemplateKMatrix}\ ,\qquad\quad
 E_{n,i}^{(2)}: =\raisebox{-1.925cm}{\PlanarTemplateRMatrix}\ ,
\ee
where $j=1,\ldots,n$ and $i=1,\ldots,n-1$, respectively. 
Similar traces and partial traces can be defined by `closing' on the left; however, for our purposes, 
we only need the right-closing (partial) traces.
Using the last diagram in (\ref{TnuJaws}), we now obtain
\begin{align}\label{equ:CAtransferAlg}
        T_{n}(u)=\Pb_{\tau_{n+1}}\big(R_{n}(u)
        \cdots R_{1}(u)K_1(u)R_{1}(u)\cdots R_{n}(u)\Ko_{n+1}(u)\big), 
\end{align}
where 
\be
 K_1(u): =\Pb_{E_{n+1,1}^{(1)}}(K(u)), \qquad
 R_{i}(u): =\Pb_{E_{n+1,i}^{(2)}}(R(u)), \qquad
 \Ko_{n+1}(u): =\Pb_{E_{n+1,n+1}^{(1)}}(\Ko(u)).
\ee

\subsection{Crossing symmetry}
\label{Sec:Cross}

For each $n\in\Nbb$, we introduce the {\em rotation tangles}
\begin{align}\label{equ:RotTan}
    r_{n,1}:= \raisebox{-1.2cm}{\RotationPlanarAntiClock},\qquad 
    r_{n,0}:= \raisebox{-1.2cm}{\RotationPlanarZero}\ ,\qquad
    r_{n,-1}:= \raisebox{-1.2cm}{\RotationPlanarClock},
\end{align}
where the numbers of spokes are such that $\eta(D^{r_{n,\pm 1}}_0) = \eta(D^{r_{n,0}}_0) = n $, 
and let $r_{n,\pm k}$, $k\in\Nbb$, denote the composition of $r_{n,\pm 1}$ with itself $k$ times. Accordingly,
\be
 r_{n,k}\circ r_{n,l}=r_{n,k+l},\qquad\quad \forall\,k,l\in\Zbb,
\ee
and we note that $r_{n,l} = r_{n,l\:\mathrm{mod}\:n}$ for all $l\in\Zbb$.
Naturality ensures that $\mathsf{P}_{r_{2n,k}}\colon A_n\to A_n$ is invertible for all $n,k$, 
see Appendix \ref{Sec:Unitality} for details.

We say that the $K$- and $R$-operators are {\em crossing symmetric} if
\begin{align}
    \mathsf{P}_{r_{2,1}}(K(u)) = \tilde{c}_K(u) K(c_K(u)),\quad 
    \mathsf{P}_{r_{4,1}}(R(u)) = \tilde{c}_R(u) R(c_R(u)),\quad
    \mathsf{P}_{r_{2,1}}(\Ko(u)) = \tilde{c}_{\Ko}(u) \Ko(c_{\Ko}(u)),
\end{align}
for some scalar functions $\tilde{c}_K,c_K,\tilde{c}_R,c_R,\tilde{c}_{\Ko},c_{\Ko}:\Omega\to\Fbb$ such that 
$\mathsf{P}_{r_{2,2}}(K(u)) = K(u)$, $ \mathsf{P}_{r_{4,4}}(R(u)) = R(u)$ and 
$\mathsf{P}_{r_{2,2}}(\Ko(u)) = \Ko(u)$. If the relation for $R$ only holds with $r_{4,1}$ replaced by $r_{4,2}$, 
the $R$-operator is considered {\em partially} crossing symmetric.
We say $u_{iso}\in\Fbb$ is an {\em isotropic point} if
\begin{align}
    \mathsf{P}_{r_{2,1}}(K(u_{iso})) = K(u_{iso}),\qquad 
    \mathsf{P}_{r_{4,1}}(R(u_{iso})) = R(u_{iso}),\qquad
    \mathsf{P}_{r_{2,1}}(\Ko(u_{iso})) =\Ko(u_{iso}).
\label{iso}
\end{align}

\noindent
\textbf{Remark.}
{\em Rigid planar algebras} \cite{Burns11} use {\em rigid planar isotopy} 
in place of the (full) planar isotopy inherent in the planar algebras in Section \ref{Sec:PlanarBasics},
essentially discarding the rotation tangles $r_{n,\pm1}$ in (\ref{equ:RotTan}).
Rigid planar algebras thus provide a natural algebraic framework for describing models without rotation symmetry, 
such as the $A_2^{(1)}$ models discussed in \cite{MDPR19}.

\subsection{Baxterisation and integrability}
\label{Sec:Baxter}

A model described by the transfer operator $T_n(u)$ is {\em integrable} on $\Omega$ if
\be
 [T_n(u),T_n(v)]=0,\qquad \forall\, u,v\in\Omega,
\label{TnTn}
\ee
with $\Omega\subseteq\Fbb$ a suitable domain, and we say that the model is {\em Yang--Baxter integrable} 
if the $R$- and $K$-operators satisfy a set of (local) relations that imply (\ref{TnTn}).
Following \cite{BaxterizationJones1990}, the associated parameterisations (\ref{RKK}) are then 
called a {\em Baxterisation}. As given, however, the functions in (\ref{RKK}) are unspecified, so a natural goal
is to constrain them to ensure Yang--Baxter integrability.
As we allow for {\em different} parameterisations of the auxiliary operators in each of the Yang--Baxter equations below,
our considerations are more general than what is usually done in the literature.

The following proposition is formulated in the diagrammatic representation used in (\ref{RKK2})
and (\ref{TnuJaws}) but is readily reformulated in planar-algebraic terms. Importantly, each of the relations in 
(\ref{Invs})--(\ref{BYBEs}) is {\em local} in the sense that there exists an `ambient planar tangle' with a suitable 
marking, relative to which it holds. 
In fact, the invertibility of the linear maps $\mathsf{P}_{r_{n,\pm k}}$ associated with the rotation tangles in Section \ref{Sec:Cross}, 
implies that the specific marking of the ambient planar tangle is immaterial. To illustrate, we have
\be
    \raisebox{-1.2cm}{\RotationPlanarZeroVecU} = \raisebox{-1.2cm}{\RotationPlanarZeroVecV} 
    \qquad \Longleftrightarrow \qquad 
    \raisebox{-1.2cm}{\RotationPlanarZeroVecURot} = \raisebox{-1.2cm}{\RotationPlanarZeroVecVRot}\,,
\ee
where $a,b\in A_n$, with the equalities statements in $A_n$.
\begin{prop}
\label{Prop:YBE}
Let the parameterisations in (\ref{RKK}) be given, and suppose there exist
\begin{align}
 \raisebox{-0.325cm}{\RotRopDarkRedOpBBi}\; :=\sum_{a\in B_2} \yo^{(i)}_a(u,v)\,a,\qquad\quad
 &\raisebox{-0.325cm}{\RotRopRedOpBBi}\; :=\sum_{a\in B_2} y^{(i)}_a(u,v)\,a,
\label{XY}
\end{align}
where $\yo^{(i)}_a$ \!and\, $y^{(i)}_a$, $i=1,2,3$, are scalar functions $\Omega\times\Omega\to\Fbb$, 
such that the following three sets of relations are satisfied:
\begin{enumerate}
\item[$\bullet$] Inversion identities (Inv1 - Inv3)
\begin{align}
  \raisebox{-0.45cm}{\RotMarkRedInvrFilteredi}\ =\; \raisebox{-0.45cm}{\MarkRedIdFilteredInv}
  \qquad\quad (i=1,2,3)
\label{Invs}
\end{align}
\item[$\bullet$] Yang--Baxter equations (YBE1 - YBE3) 
\be
 \raisebox{-0.93cm}{\RotMarkRedYBEgoLeftSingleRia}\ =\raisebox{-0.93cm}{\RotMarkRedYBEgoRightSingleRia}
\qquad\quad
 \raisebox{-0.93cm}{\RotMarkRedYBEooRightSingleRib}\ =\raisebox{-0.93cm}{\RotMarkRedYBEooLeftSingleRib} 
\qquad\quad
 \raisebox{-0.93cm}{\RotMarkRedYBEggLeftSingleRic}\ = \raisebox{-0.93cm}{\RotMarkRedYBEggRightSingleRic}
\label{YBE}
\ee
\item[$\bullet$] Boundary Yang--Baxter equations (BYBEs)
\begin{align}
 \raisebox{-1.43cm}{\RotMarkRedBYBEbFilteredLeft}\ \;=\; \raisebox{-1.43cm}{\RotMarkRedBYBEbFilteredRight} 
\qquad\qquad\qquad
 \raisebox{-1.43cm}{\RotMarkRedBYBEyFilteredLeft}\ \;=\; \raisebox{-1.43cm}{\RotMarkRedBYBEyFilteredRight}
\label{BYBEs}
\end{align}
where
\begin{align}
 \raisebox{-0.325cm}{\RotRopDarkRedOpBBid}\; =\sum_{a\in B_2} \yo^{(1)}_a(v,u)\,a,\qquad\quad
 &\raisebox{-0.325cm}{\RotRopRedOpBBid}\; =\sum_{a\in B_2} y^{(1)}_a(v,u)\,a.
\label{YBE4}
\end{align}
\end{enumerate}

Then, $[T_n(u),T_n(v)]=0$ for all $u,v\in\Omega$.
\end{prop}
\begin{proof}
Using the following familiar manipulations \cite{BPOB96},
\begin{align}
    T_n(u)T_n(v) &= \raisebox{-1.95cm}{\RotMarkRedIntFilteredISingleR} 
    \stackrel{\mathclap{\scalebox{0.5}{\normalfont\mbox{(Inv1)}}}}{=} 
    \raisebox{-1.95cm}{\RotMarkRedIntFilteredIISingleR} 
    \stackrel{\mathclap{\scalebox{0.5}{\normalfont\mbox{(YBE1)}}}}{=} 
    \raisebox{-1.95cm}{\RotMarkRedIntFilteredIIISingleR}
\nonumber
\end{align}
    \vspace{-0.25cm}
\begin{align}
    &\stackrel{\mathclap{\scalebox{0.5}{\normalfont\mbox{(Inv2)}}}}{=} 
    \raisebox{-1.95cm}{\RotMarkRedIntFilteredIVSingleR} \;\;
    \stackrel{\mathclap{\scalebox{0.5}{\normalfont\mbox{(YBE2)}}}}{=}\;
    \raisebox{-1.95cm}{\RotMarkRedIntFilteredVSingleR}
\nonumber
\end{align}
    \vspace{-0.25cm}
\begin{align}
    &\stackrel{\mathclap{\scalebox{0.5}{\normalfont\mbox{(BYBEs)}}}}{=}\;  
    \raisebox{-1.95cm}{\RotMarkRedIntFilteredVISingleR} \;
    \stackrel{\mathclap{\scalebox{0.5}{\normalfont\mbox{(YBE3)}}}}{=}  
    \raisebox{-1.95cm}{\RotMarkRedIntFilteredVIISingleR}
\nonumber
\end{align}
    \vspace{-0.25cm}
\begin{align}
    &\stackrel{\mathclap{\scalebox{0.5}{\normalfont\mbox{(Inv3+YBE4)}}}}{=}\;\;\;  
    \raisebox{-1.95cm}{\RotMarkRedIntFilteredVIIISingleR} 
    \stackrel{\mathclap{\scalebox{0.5}{\normalfont\mbox{(Inv4)}}}}{=}  
    \raisebox{-1.95cm}{\RotMarkRedIntFilteredIXSingleR} 
    = T_n(v)T_n(u),
\label{4}
\end{align}
we arrive at the result. In (\ref{4}), YBE4 and Inv4 refer respectively to YBE1 and Inv1 with $u$ and $v$ 
interchanged and $1$ replaced by $4$, c.f. (\ref{YBE4}).
\end{proof}
\noindent
\textbf{Remark.} 
Proposition \ref{Prop:YBE} provides {\em sufficient} conditions for integrability; they are in general 
not necessary. The elements in (\ref{XY}) depend on $u$ and $v$, but their parameterisations may be different 
from the ones of $R$. We therefore say that the YBEs (\ref{YBE}) are {\em generalised}. 
While the work \cite{PRVO20} on dimers also relies on more than one YBE to establish commutativity of transfer 
operators, the auxiliary operators playing the role of (\ref{XY}) are all used as building blocks in the construction of 
the dimer transfer operator. We stress that the auxiliary operators in the generalised YBEs (\ref{YBE}) do not 
necessarily appear in the transfer operator; they are in many ways a means to an end in establishing 
commutativity.
\medskip

\noindent
\textbf{Remark.}
Relations like the BYBEs (\ref{BYBEs}) are often referred to as {\em reflection equations}.

\subsection{Sklyanin's formulation}
\label{Sec:Sklyanin}

The partial trace appearing in (\ref{equ:CAtransferAlg}) is diagrammatic in origin.
In fact, for a given planar algebra, there need not be a vector space over which the trace may be interpreted 
as being performed. However, if the algebra and the corresponding $R$-operator satisfy certain conditions, 
we {\em can} identify such an auxiliary vector space. 

For each $m,n\in\Nbb_0$, the quadratic tangle
\begin{align}\label{equ:CoMultPlanar}
    K_{m,n} := \raisebox{-1.4275cm}{\RotPlanarIIKCoMultAltC},\qquad\quad
    \Pb_{K_{m,n}}\colon  A_m\times A_n\to A_{m+n},
\end{align}
induces a \textit{tensor product} between $A_m$ and $A_n$ within $A_{m+n}$.
For ease of notation, for $a\in A_{m}$ and $b\in A_{n}$, 
we write $a\otimes b=\Pb_{K_{m,n}}(a,b)\in A_{m+n}$. 
We say that $R(u)\in A_2$ is \textit{separable} if it can be written as an element of $A_1\otimes A_1$.
Diagrammatically, this amounts to a decomposition of the form
\begin{align}\label{equ:SeparableRStraight}
    R(u) = \raisebox{-0.45cm}{\RotRopGreenOpBBV} 
    = \sum_{a_1,a_2\in B_1}\!\!\!\!\!r_{a_1,a_2}(u) \raisebox{-0.45cm}{\MarkRedRiFilteredCrossRevSepAlt}.
\end{align}
The auxiliary vector space of the corresponding transfer operator is then the rightmost {\em channel}, 
here coloured blue:
\begin{align}\label{equ:OrthTransBlue}
    T_n(u) = \raisebox{-0.95cm}{\RotMarkRedTransferFilteredCrossIsoSingleRALT}.
\end{align}
We also say that a planar algebra is \textit{braided}, respectively \textit{symmetric}, 
if each vector space $A_n$ admits a representation of the $n$-strand braid, respectively symmetric, 
group algebra. In the symmetric case, applying the permutation element of $A_2$ to the separable 
$R$-operator \eqref{equ:SeparableRStraight}, we obtain
\begin{align}\label{equ:SeparableRTwist}
    \check{R}(u)= \raisebox{-0.45cm}{\RotRopGreenOpBBVcheck}
    := \sum_{a_1,a_2\in B_1}\!\!\!\!\!r_{a_1,a_2}(u) \raisebox{-0.5cm}{\VecMatAlRopDecompStrandsTwistIGen}
\end{align}
and subsequently
\begin{align}
   \check{T}_n(u) = \raisebox{-0.95cm}{\RotTransferFilteredCrossIsoSingleRV},
\label{TS}
\end{align}
corresponding to the usual double-row transfer matrix in Sklyanin's formulation \cite{Sklyanin88}.
By combining the $R$-operators \eqref{equ:SeparableRStraight} and \eqref{equ:SeparableRTwist}, one can 
construct transfer operators, where the auxiliary space is `threaded' through any of the intermediate channels
\begin{align}\label{equ:SkyIntermedTransferOps}
    \raisebox{-0.95cm}{\RotTransferFilteredCrossIsoSingleRI}, \;\; \ldots, 
    \raisebox{-0.95cm}{\RotTransferFilteredCrossIsoSingleRII}.
\end{align}
Note that the commutativity of the operators in \eqref{equ:OrthTransBlue}, and likewise of the ones in \eqref{TS},
does not necessarily imply commutativity of any of the intermediate operators 
in \eqref{equ:SkyIntermedTransferOps}.

\subsection{Hamiltonian limits}
\label{Sec:HamiltonianSub}

We say that $u_*\in\Fbb$ is an {\em identity point} of $T_n(u)$ if $T_n(u_*)$ is proportional to $\idn$.
Around each identity point for which the proportionality constant is nonzero, we will use a power-series expansion 
of the transfer operator to define associated {\em hamiltonian operators}.

In preparation, for $\Omega\subseteq\Fbb$ a suitable domain, we define
\begin{align}
        g\colon \Omega\to A_0,\qquad u \mapsto \raisebox{-0.45cm}{\RotMarkRedTransferFilteredCrossIsoUV}\ ,
\end{align}
with the image written in our usual diagrammatic representation.
Note that
\be
 g(u)=\Pb_{\tr_1\circ M_1}\big(K(u),\Ko(u)\big)=\Pb_{\tr_1\circ M_1}\big(\Ko(u),K(u)\big).
\ee
Composing $g$ with the evaluation map, we obtain the scalar function
\be
 \gh: =\mathrm{e}\circ g\colon \Omega\to\Fbb.
\ee

The following proposition provides sufficient conditions for the existence of identity points.
\begin{prop}\label{prop:HamLimitPointsGen}
Let $u_*\!\in\Omega$ and suppose there exist $l_{u_*},r_{u_*}\!\in\Fbb$ such that
\begin{align}
\label{lr}
  \raisebox{-1cm}{\RotMarkRedIdPointFilteredBlueLeftSingleR} 
  = l_{u_*}\raisebox{-1cm}{\RotMarkRedIdPointFilteredBlueRight}
  \qquad\ \text{or}\qquad\
   \raisebox{-1cm}{\RotMarkRedIdPointFilteredYellowLeftSingleR} 
   = r_{u_*}\raisebox{-1cm}{\RotMarkRedIdPointFilteredYellowRight}\ .
\end{align}
Then, $u_*$ is an identity point, with 
$T_n(u_*)=l^n_{u_*}\gh(u_*)\idn$ or $T_n(u_*)=r^n_{u_*}\gh(u_*)\idn$, respectively. 
\end{prop}
\begin{proof}
If the left relation in (\ref{lr}) holds, then
\begin{align}
        T_n(u_*) = \raisebox{-0.925cm}{\RotMarkRedTransferFilteredCrossIsoUstarISingleR} 
        = l_{u_*}\raisebox{-0.925cm}{\RotMarkRedTransferFilteredCrossIsoUstarIISingleR}
        = l^n_{u_*}\raisebox{-0.925cm}{\RotMarkRedTransferFilteredCrossIsoUstarIII} 
        = l^n_{u_*}\gh(u_*)\idn.
\end{align}
Similar arguments apply if the right relation holds.
\end{proof}
\noindent
\textbf{Remark.} If both relations in (\ref{lr}) hold and $\gh(u_*)\neq0$, 
then $l_{u_*}^{n-k}r_{u_*}^k=l_{u_*}^{n-k'}r_{u_*}^{k'}$ for all $k,k'\in\{0,1,\ldots,n\}$, hence $l_{u_*}=r_{u_*}$.
\medskip

Now, let $u_*$ be an identity point, and suppose $\Omega$ contains an open subset of $\mathbb{F}$ containing $u_*$.
If $T_n(u)$ is not just proportional to $\idn$, 
then there exist unique $k\in\Nbb$ and unique $H_{n,u_*}\notin\Fbb\idn$ such that
\be
   T_n(u_*+\epsilon)=p_{k-1}(\epsilon)\idn + \epsilon^k H_{n,u_*} + \Oc(\epsilon^{k+1}),
\label{equ:FilteredHamFromTrans}
\ee
where $p_{k-1}$ is a polynomial of degree at most $k-1$. Since $T_n(u_*)\neq0$,
we have $p_{k-1}(0)\neq0$, and if $u_*$ is of the form in Proposition \ref{prop:HamLimitPointsGen}, 
then $p_{k-1}(0)=l^n_{u_*}\gh(u_*)$ or $p_{k-1}(0)=r^n_{u_*}\gh(u_*)$.
We refer to $H_{n,u_*}$ as the {\em hamiltonian} associated with the given identity point
(and given transfer operator), and note that
\be
  H_{n,u_*} =\frac{1}{k!} \frac{\partial^k}{\partial u^k}T_n(u)\big|_{u=u_*}.
\label{equ:HamDerivativeFormFiltered}
\ee
By construction, there exist $s_k\in\Fbb$ and nonzero $h_{n,u_*}\in\mathrm{span}_{\Fbb}(B'_n)$ such that
\be
 H_{n,u_*}=s_k\idn+h_{n,u_*}.
\ee
Setting
\be
 \tilde{p}_k(\epsilon):=p_{k-1}(\epsilon)+s_k\epsilon^k,
\ee
we obtain
\be
 T_n(u_*+\epsilon)=\tilde{p}_k(\epsilon)\idn+\epsilon^kh_{n,u_*}+ \Oc(\epsilon^{k+1}),
\label{hprincipal}
\ee
where $\tilde{p}_k$ is a polynomial of degree at most $k$. Up to a possible rescaling, we refer to
$h_{n,u_*}$ as the {\em principal hamiltonian} associated with the identity point $u_*$.
If confusion is unlikely to occur, we may write $h_{u_*}$ or $h_n$ instead of $h_{n,u_*}$.
Different identity points associated with $T_n(u)$ may give rise to different principal hamiltonians.
Hamiltonians arising from different Baxterisations are also likely to correspond to different elements of $A_n$.
\medskip

\noindent
\textbf{Remark.} 
If $u_*$ is such that (\ref{lr}) holds and $T_n(u_*)=0$, then one may renormalise the original transfer operator
such that the limit $u\to u_*$ yields a well-defined nonzero operator, see e.g.~\cite{PR07}.
This will be illustrated in Section \ref{Sec:h}.
\medskip

\noindent
If $T_n(u)$ describes an integrable model on $\Omega$, then we have the familiar relations
\be
 [h_{n,u_*},T_n(u)]= [H_{n,u_*},T_n(u)]=0,\qquad \forall\,u_*,u\in\Omega.
\ee

\subsection{Hamiltonians and integrals of motion}
\label{Sec:Higher}

Let $\Ac$ be an associative algebra with basis $\Bc$, 
and let $T(u)\in\Ac$ denote a transfer operator of some model. Then,
\be
 T(u)=\sum_{a\in\Bc}t_a(u)\,a,
\ee
where $t_a\colon\Omega\to\Fbb$ for each $a\in\Ac$, with $\Omega\subseteq\Fbb$ a suitable domain. 
As a space of scalar functions, we have
\be
 \Fc: =\mathrm{span}_\Fbb\,\{t_a\colon\Omega\to\Fbb\,|\,a\in\Bc\}
\ee
and let $\Bc_T$ denote a basis for $\Fc$. Relative to this, we can decompose the transfer operator as
\be
 T(u)=\sum_{f\in\Bc_T}f(u)\,a_f,
\ee
where $a_f\in\Ac$ for each $f\in\Bc_T$. We also introduce
\be
 \Hc_T: =\mathrm{span}_\Fbb\,\{a_f\,|\,f\in\Bc_T\}
\ee
and note that $\dim\Hc_T\leq\dim\Fc$.
Related to this, we have the $\Ac$-subalgebra
\be
 \Ac_T: =\langle\Hc_T\rangle_\Ac.
\ee 

Any nonzero element $h\in\Hc_T$ (such that $h\notin\Fbb\mathds{1}$ for $\Ac$ unital) 
could conceivably be regarded as the hamiltonian of the model, 
although the physics would likely guide the selection.
The hamiltonians $H_{n,u_*}$ in (\ref{equ:FilteredHamFromTrans}) are examples of such algebra elements.
Here, we refer to the elements of $\Hc_T$ generally as {\em hamiltonians}
and to certain preferred choices, $h\in\Hc_T$, as the {\em principal hamiltonians}.
In general distinct from this notion of hamiltonians, we view the centraliser of $h$ in $\Ac$, here denoted by 
$\Cc_\Ac(h)$, as the space (in fact, subalgebra) of $h$-\textit{conserved quantities}.

If the model is integrable in the sense that
\be
 [T(u),T(v)]=0,\qquad \forall\, u,v\in\Omega,
\ee
then $\Ac_T$ is commutative, and given a choice of $h$, every hamiltonian is an $h$-conserved quantity.
The converse need not be true; that is, an $h$-conserved quantity need not be a hamiltonian of the model, 
as we generally only have the inclusion $\Hc_T\subseteq\Cc_\Ac(h)$.
Also, noting that $\Cc_\Ac(\Hc_T)$ is a subalgebra of $\Ac$, we view its centre, $\Zc(\Cc_\Ac(\Hc_T))$, 
as the subalgebra of {\em integrals of motion of the model}. 
In summary, these various algebras of the integrable model are related as
\be
 \Ac_T\subseteq\Zc(\Cc_\Ac(\Hc_T))\subseteq\Cc_\Ac(\Hc_T)\subseteq\Cc_\Ac(h)\subseteq\Ac.
\ee

\subsection{Polynomial integrability}
\label{Sec:Pol}

Let $\Ac$ be an associative algebra and $\{T(u)\,|\,u\in\Omega\}\subseteq\Ac$ a one-parameter family of 
commuting elements, with $\Omega\subseteq\Fbb$. If there exists $b\in\Ac$ such that, for all $u\in\Omega$,
\be
 T(u)\in\Fbb[b],
\label{Tub}
\ee
then $[T(u),T(v)]=0$ is trivially satisfied. If the commuting family corresponds to the transfer operators of some 
model, e.g.~$\{T_n(u)\,|\,u\in\Omega\}\subseteq A_n$ in the planar-algebraic setting above, and 
satisfies (\ref{Tub}) for some $b_n\in A_n$, then we say that the model is {\em polynomially integrable}. 
A polynomially integrable model may thus be seen as {\em trivially integrable}.

It may seem surprising at first that a nontrivial physical model could be polynomially integrable,
but that is exactly what we find. In fact, a hamiltonian operator may play the role of $b$ in (\ref{Tub}),
in which case {\em the transfer operators are polynomial in the hamiltonian}. 
Examples of this will be provided in the following.

\section{Endomorphism algebras}
\label{Sec:Linear}

Having outlined our integrability framework for planar algebras in Section \ref{Sec:Planar}, we digress with 
(i) an examination of parameterised families of linear operators, introducing the notion of {\em spurious} 
spectral degeneracies, of which there can only be finitely many,
(ii) a discussion of parameter-dependent matrices (respectively algebra generators) expressible as polynomials 
in a parameter-independent matrix (respectively algebra generator),
and (iii) a brief review of some basic properties of cellular algebras.
A key result of this section is Proposition \ref{prop:Asemi} which gives necessary and sufficient conditions 
for a transfer operator to be polynomial in a single integral of motion, c.f.~\eqref{Tub}.
Other results include Proposition \ref{prop:SpDegenAx} and 
Corollary \ref{Cor:diagsemi}, which are used in Section \ref{Sec:PolTL} and Section \ref{Sec:TLB}, respectively.
In the following, we specialise $\Fbb$ to $\Cbb$.

\subsection{Spectral degeneracies}
\label{Sec:Spectral}

Let $A(x)$ be an $x$-dependent linear map $\Cbb^n\to\Cbb^n$ that admits a matrix representation
where every matrix element is polynomial in $x$, and denote the associated characteristic polynomial by
\be
 c(x,\lambda): =\det(\lambda\,\id-A(x)).
\ee
By construction, $c(x,\lambda)$ is polynomial in $x$ and $\lambda$, 
and the eigenvalues $\lambda$ will in general depend on $x$.
If the spectrum of $A(x)$ for $x$ an indeterminate possesses fewer than $n$ distinct eigenvalues, 
we say that the corresponding spectral degeneracies are \textit{permanent}. 
Let $l$ denote the number of distinct eigenvalues for $x$ an indeterminate.
If $A(x_0)$ has fewer than $l$ distinct eigenvalues for some value $x_0\in\Cbb$, 
we refer to the additional, non-permanent spectral degeneracies as \textit{spurious}.
As Proposition \ref{prop:SpDegenAx} below asserts, 
there are only finitely many values $x_0$ for which spurious degeneracies arise. 
Here and elsewhere, we allow ``finitely many" to mean zero.

First, recall that a polynomial $f(x,y)$ is said to be \textit{irreducible} if it cannot be expressed as the product of 
two non-constant polynomials.
For such a polynomial, the equation $f(x,y)=0$ determines an algebraic function $y(x)$, and we recall that
an algebraic function has a finite number of branches and at most algebraic singularities.

Also, two polynomials $f(x,y)$ and $g(x,y)$ are \textit{relatively prime} if they do not share a common 
non-constant polynomial factor. In this case, there are only finitely many values $x_0$ 
for which the $y$-polynomials $f(x_0,y)$ and $g(x_0,y)$ have a common root 
(see for example Theorem 3 on page 300 of \cite{Ahlfors79}). 
We can now establish the following result which is used in Section \ref{Sec:PolTL}.
\begin{prop}\label{prop:SpDegenAx}
The spectrum of $A(x)$ is non-spurious for all but finitely many $x\in\Cbb$.
\end{prop}
\begin{proof}
The characteristic polynomial of $A(x)$ admits a factorisation of the form
\be
 c(x,\lambda) = \prod_{i=1}^{t}c_i(x,\lambda),
\ee
where $c_i(x,\lambda)\in\Cbb[x,\lambda]$ is irreducible for every $i=1,\ldots,t$, and where $1\leq t\leq n$.
For $i\neq j$, $c_i(x,\lambda)$ and $c_j(x,\lambda)$ are relatively prime or scalar multiples of each other.
In the latter case, the pairing only contributes permanent degeneracies to the spectrum of $A(x)$.
Because each pair of relatively prime factors (of which there are finitely many) 
will contribute finitely many spurious degeneracies, and since each of the $t$ irreducible factors 
individually can contribute finitely many spurious degeneracies, the result follows.
\end{proof}

\subsection{Block Toeplitz}
\label{Sec:Toeplitz}

We let $\Mc_n(\Rc)$ denote the set of $n\times n$ matrices with elements from $\Rc$. 
For each $z\in\Mc_n(\Cbb)$, we denote the centraliser of $z$ in $\Mc_n(\Cbb)$ by $\Cc(z)$ and let $c_z$ 
and $m_z$ denote the characteristic and minimal polynomial, respectively. Recall that $z$ is {\em non-derogatory} 
if $c_z=m_z$. Equivalently, it is non-derogatory if and only if every eigenspace of $z$
is at most one-dimensional; and if and only if $z$ generates its own centraliser, that is, $\Cc(z)=\Cbb[z]$,
see e.g.~\cite{NP95}. Note that if $z\in\Mc_n(\Cbb)$ is similar to a non-derogatory matrix, then $z$ is 
non-derogatory. We say that a {\em linear operator} on $\Cbb^n$ is non-derogatory if there exists a basis relative 
to which its matrix representation is non-derogatory.

A matrix $J$ is in \textit{Jordan canonical form} (JCF) if it is block-diagonal with Jordan blocks on the diagonal, 
$J=\mathrm{diag}(J_{r_1}(\lambda_1),\ldots,J_{r_s}(\lambda_s))$,
where
\begin{align} 
    J_r(\lambda) : = 
    \left[
    \begin{array}{ccccc}
         \lambda & 1 & 0 & \ldots & 0\\
         0 & \lambda & 1  & \ddots &\vdots\\
         \vdots & \ddots & 
         \ddots  & \ddots & 0 \\ 
         0 & \ldots & 0 & \lambda & 1\\
         0 & \ldots & 0 & 0 & \lambda
    \end{array}
    \right]\in\Mc_r(\Rc).
\end{align}
For $f$ an analytic function, we then have 
$f(J) = \mathrm{diag}\left(f(J_{r_1}(\lambda_1)),\ldots,f(J_{r_s}(\lambda_s))\right)$, where
\begin{align}\label{equ:JBFunction}
    f\left(J_r(\lambda)\right) = 
    \left[
    \begin{array}{ccccc}
         f(\lambda) & f'(\lambda) & \frac{f''(\lambda)}{2} & \ldots & \frac{f^{(r-1)}(\lambda)}{(r-1)!}\\
         0 & f(\lambda) & f'(\lambda)  & \ddots
         &\vdots\\
         \vdots& \ddots
         & \ddots  & \ddots
         & \frac{f''(\lambda)}{2} \\[.1cm]
         0 & \ldots & 0 & f(\lambda) & f'(\lambda)\\[.1cm]
         0 & \ldots & 0 & 0 & f(\lambda)
    \end{array}
    \right].
\end{align}

An upper-triangular Toeplitz matrix in $\Mc_r(\Rc)$ is characterised by a tuple in 
$\Rc^{\times r}$, $\ab_r= (a_1, \ldots, a_r)$, as
    \begin{align}\label{equ:Tdef}
    T(\ab_r)= 
    \left[
    \begin{array}{cccc}
         a_1 & a_2 & \ldots & a_{r}\\
         0 & a_1 & \ddots  & \vdots\\[-.1cm]
         \vdots& \ddots
         & \ddots & a_2  \\ 
         0 & \ldots & 0 & a_1
    \end{array}
    \right].
    \end{align}
We refer to a matrix of the form
\be
 \mathrm{diag}\left(T(\mathbf{a}_{r_1}^{[1]}), \ldots,T(\mathbf{a}_{r_s}^{[s]}) \right)
\ee
as \textit{block-diagonal upper-triangular Toeplitz} (BT) and say it has {\em block-partitioning} $r_1,\ldots,r_s$.
Note that a block-partitioning need not be unique. Indeed, the $2\times2$ identity matrix admits
the block-partitionings $2$ and $1,1$.
\begin{lem}\label{lem:Bx}
Let $B(x)\in\Mc_n(\Cbb(x))$.
Then, there exists $b\in\Mc_n(\Cbb)$ in JCF such that $B(x)\in\Cbb(x)[b]$ if and only if $B(x)$ is BT.
\end{lem}
\begin{proof}
As to ``$\Rightarrow$'', we have $B(x)=p_x(b)$, where $p_x$ is a polynomial with $x$-dependent coefficients, and
\begin{align}
        b = \mathrm{diag}\left(J_{r_1}(\lambda_1),\ldots,J_{r_s}(\lambda_s)\right)
\end{align}
for some $r_1,\ldots,r_s\in\Nbb$ and $\lambda_1,\ldots,\lambda_s\in\Cbb$.
By \eqref{equ:JBFunction},
\begin{align}
        B(x) = \mathrm{diag}\left(p_x(J_{r_1}(\lambda_1)),\ldots,p_x(J_{r_s}(\lambda_s))\right),
\end{align}
and since each $p_x(J_{r_i}(\lambda_i))$ is an upper-triangular Toeplitz matrix, it follows that $B(x)$ is BT. 
As to ``$\Leftarrow$'', let
\begin{align}\label{bdiag}
        b = \mathrm{diag}\left(J_{r_1}(\lambda_1),\ldots,J_{r_s}(\lambda_s)\right),
\end{align}
where $r_1,\ldots,r_s$ is a block-partitioning of $B(x)$ and where $\lambda_1,\ldots,\lambda_s\in\Cbb$ are all 
distinct. Since $b$ is non-derogatory, it generates its own centraliser. Because of the shared block-partitioning, 
we have $[b,B(x)]=0$. It follows that $B(x)\in\Cbb(x)[b]$.
\end{proof}
\noindent
\textbf{Remark.} The proof of Lemma \ref{lem:Bx} holds for {\em any} block-partitioning of $B(x)$. 
Indeed, the key properties of the matrix $b$ constructed in (\ref{bdiag}) are that it is non-derogatory and 
that $[b,B(x)]=0$.
\begin{prop}\label{prop:BxMn}
Let $B(x)\in\Mc_n(\Cbb(x))$. Then, there exists $b\in\Mc_n(\Cbb)$ such that $B(x)\in\Cbb(x)[b]$
if and only if there exists $S\in\Mc_n(\Cbb)$ such that $S^{-1}B(x)S$ is BT.
\end{prop}
\begin{proof}
As to ``$\Rightarrow$'', we have $B(x)=p_x(b)$, where $p_x$ is a polynomial with $x$-dependent coefficients,
and let $S\in\Mc_n(\Cbb)$ such that $S^{-1}bS$ is in JCF. 
Then, $S^{-1}B(x)S=S^{-1}p_x(b)S=p_x(S^{-1}bS)$, so by Lemma \ref{lem:Bx}, $S^{-1}B(x)S$ is BT.
As to ``$\Leftarrow$'', let $S^{-1}B(x)S$ be BT. Then, by Lemma \ref{lem:Bx}, there exists $\bar{b}\in\Mc_n(\Cbb)$
such that $S^{-1}B(x)S\in\Cbb(x)[\bar{b}]$. It follows that $B(x)\in\Cbb(x)[b]$, where $b=S\bar{b}S^{-1}$.
\end{proof}
\begin{cor}\label{cor:Cxdiag}
Let $\{C(x)\in\Mc_n(\Cbb)\,|\,x\in\Omega\}$ be a one-parameter family of commuting 
and diagonalisable matrices, with $\Omega\subseteq\Cbb$ a suitable domain. 
Then, there exists $b\in\Mc_n(\Cbb)$ such that $C(x)\in\Cbb[b]$ for every $x\in\Omega$.
\end{cor}
\begin{proof}
Because the matrices $C(x)$ are diagonalisable and in involution, there exists a basis for $\Cbb^n$ consisting of 
$x$-independent common eigenvectors, see e.g.~\cite{Gantmacher59}. Let $S\in\Mc_n(\Cbb)$ be constructed 
by concatenating the basis vectors. Then, for every $x\in\Omega$, $S^{-1}C(x)S$ is diagonal, hence BT, 
and the result follows from Proposition \ref{prop:BxMn}.
\end{proof}
\noindent
\textbf{Remark.} 
Corollary \ref{cor:Cxdiag} implies that {\em any} model described by a family of commuting and 
diagonalisable transfer matrices is {\em polynomially integrable}.

\medskip
As discussed in Section \ref{Sec:Ass} below, Proposition \ref{prop:BxMn} and Corollary \ref{cor:Cxdiag} elevate 
from matrices to elements of finite-dimensional unital associative semisimple algebras.

\subsection{Semisimple algebras}
\label{Sec:Ass}

Let $\Ac$ be a finite-dimensional unital associative algebra, and let $\mathrm{rad}(\Ac)$ denote its radical.
Motivated by the discussion of square matrices in Section \ref{Sec:Toeplitz},
we say that $b\in\Ac$ is {\em non-derogatory} if it generates its own centraliser, that is, if
\be
 \Cc_\Ac(b)=\langle b\rangle_\Ac.
\ee
Note that the inclusion $\Cc_\Ac(b)\supseteq\langle b\rangle_\Ac$ holds for all $b\in\Ac$.
If $b$ is non-derogatory, then any $c\in\Cc_\Ac(b)$ admits a unique expression of the form
\be
 c=\sum_{i=0}^{d-1}c_ib^i,\qquad d: =\dim\!\big(\langle b\rangle_\Ac\big),
\label{ccb}
\ee
where $c_0,\dots,c_{d-1}$ are scalars.

Recall that $\Ac$ has finitely many non-isomorphic irreducible modules, 
$L_1,\ldots,L_r$, each of which is finite-dimensional, and
\be
 \Ac/\mathrm{rad}(\Ac)\cong\bigoplus_{i=1}^r\mathrm{End}_{\Cbb}(L_i).
\ee
Letting $\rho_i$ denote the representation corresponding to $L_i$ for each $i=1,\ldots,r$,
the homomorphism
\be
 \rho: =\bigoplus_{i=1}^r\rho_i\colon \Ac\to\bigoplus_{i=1}^r\mathrm{End}_{\Cbb}(L_i)
\label{rho}
\ee
is indeed surjective with kernel given by $\mathrm{rad}(\Ac)$.

We also recall that $\Ac$ is {\em semisimple} if $\mathrm{rad}(\Ac)=\{0\}$.  Equivalently, $\Ac$ is semisimple if and 
only if
\be
 \Ac\cong\bigoplus_{i=1}^r\,(\dim L_i)\,L_i
\ee
as $\Ac$-modules; if and only if
\be
 \dim\Ac=\sum_{i=1}^r(\dim L_i)^2;
\ee
and if and only if any $\Ac$-representation, including the regular representation,
is completely reducible. 
In the following, we let
\be
 L: =\bigoplus_{i=1}^rL_i.
\ee
\begin{lem}\label{lem:nonder}
Let $\Ac$ be a finite-dimensional unital associative semisimple algebra, and let $b\in\Ac$. 
Then, $b$ is non-derogatory if and only if $\rho(b)$ is non-derogatory.
\end{lem}
\begin{proof}
As to ``$\Rightarrow$'', let $b$ be non-derogatory and $\psi\in\mathrm{End}_{\Cbb}(L)$
such that $\psi\rho(b)=\rho(b)\psi$. 
Since $\Ac$ is semisimple, $\rho$ is an isomorphism, so $\rho^{-1}(\psi)b=b\rho^{-1}(\psi)$. 
Since $b$ is non-derogatory, it follows that there exist scalars $c_i$ such that
\be
 \rho^{-1}(\psi)=\sum_ic_ib^i,
\ee
hence
\be
 \psi=\sum_ic_i\rho(b)^i,
\ee
so $\rho(b)$ is non-derogatory.
As to ``$\Leftarrow$'', let $\rho(b)$ be non-derogatory and $c\in\Cc_\Ac(b)$.
Since $bc=cb$ and $\rho$ is a homomorphism, we have $\rho(b)\rho(c)=\rho(c)\rho(b)$.
Since $\rho(b)$ is non-derogatory, it follows that there exist $\gamma_i$ such that
\be
 \rho(c)=\sum_i\gamma_i\rho(b)^i,
\ee
hence $c=\sum_i\gamma_ib^i\in\langle b\rangle_\Ac$, so $b$ is non-derogatory.
\end{proof}
We now have the following algebraic version of Proposition \ref{prop:BxMn}.
\begin{prop}\label{prop:Asemi}
Let $\Ac$ be a finite-dimensional unital associative semisimple algebra over $\Cbb$, and let $U(x)\in\Ac$.
Then, there exists $x$-independent $b\in\Ac$ such that $U(x)\in\Cbb(x)[b]$ if and only if there exists
an $x$-independent basis for $L$ such that the matrix representation of $\rho(U(x))$ is BT.
\end{prop}
\begin{proof}
As to ``$\Rightarrow$'', let $b\in\Ac$ be $x$-independent such that $U(x)\in\Cbb(x)[b]$.
Then, there exist scalars $c_i(x)$ such that
\be
 U(x)=\sum_ic_i(x)b^i.
\ee
Since $\rho$ is an homomorphism, we have
\be
 \rho(U(x))=\sum_ic_i(x)\rho(b)^i,
\ee
and since $\rho(b)$ is $x$-independent, there exists an $x$-independent $L$-basis relative to which the 
matrix representation of $\rho(b)$ is in JCF. By Lemma \ref{lem:Bx}, the corresponding matrix representation 
of $\rho(U(x))$ is BT.
As to ``$\Leftarrow$'', let $B$ be an $x$-independent $L$-basis relative to which the matrix representation, 
$M$, of $\rho(U(x))$ is BT. Construct $\psi\in\mathrm{End}_{\Cbb}(L)$ such that its matrix 
representation, $P$, relative to $B$ is in JCF with block-partitioning matching that of $M$, and with all the Jordan 
blocks having distinct eigenvalues. This matrix representation is thus non-derogatory, so by 
Lemma \ref{lem:nonder}, $b: =\rho^{-1}(\psi)$ is non-derogatory. By the construction of $\psi$, we have $PM=MP$, 
hence $\psi\rho(U(x))=\rho(U(x))\psi$, and since $\rho$ is an isomorphism, 
it follows that $b\,U(x)=U(x)b$, so $U(x)\in\Cbb(x)[b]$.
\end{proof}
\noindent
The following algebraic counterpart to Corollary \ref{cor:Cxdiag} is used in Section \ref{Sec:TLB}.
\begin{cor}\label{Cor:diagsemi}
Let  $\mathcal{A}$ be as in Proposition \ref{prop:Asemi}, $\Omega\subseteq\Cbb$ a suitable domain,
and $\{C(x)\in \mathcal{A}\,|\,x\in\Omega\}$ a commutative family of operators 
such that there exists a basis relative to which the matrix representation $\rho(C(x))$ is diagonal for all $x$. 
Then, there exists $b\in \mathcal{A}$ such that $C(x)\in\Cbb[b]$ for every $x\in\Omega$.
\end{cor}
\noindent
\textbf{Remark.}
{\em Subfactor planar algebras} \cite{Jones99,JonesNotes} are a class of planar algebras that present themselves 
in the context of polynomial integrability. They are semisimple and possess a natural inner product, so establishing 
{\em diagonalisability} of the transfer operator in a given representation amounts to showing that the operator is 
self-adjoint with respect to the inner product. In fact, the {\em global} property of polynomial integrability is reduced 
to satisfying the {\em local} YBEs, BYBEs, inversion identities, and self-adjointness of the $R$- and $K$-operators.

\subsection{Cellular algebras}
\label{Sec:Cellular}

Here, we recall the definition and some basic properties of cellular algebras to be used in 
Section \ref{Sec:Cellularity} and Appendix \ref{app:TensorPADef}. 
Proofs will not be reproduced but may be found in \cite{GL06}.
The following definition is taken from \cite{GL06}.
\begin{defn}\label{def:CellAlg}
A cellular algebra over $\Rc$ is an associative unital algebra $\Ac$, together with cell datum $(\Lambda,M,C,*)$, 
where 
\begin{enumerate}
\item[$\bullet$]
$\Lambda$ is a partially ordered set and for each $\lambda\in\Lambda$, $M(\lambda)$ is a finite set such that %
\begin{align}
  C\colon \coprod_{\lambda\in\Lambda} M(\lambda)\times M(\lambda) \to\Ac
\end{align}
is an injective map with image an $\Rc$-basis for $\Ac$.
\item[$\bullet$]
If $\lambda\in\Lambda$ and $S,T\in M(\lambda)$, write $C(S,T) = C^{\lambda}_{S,T}$.
Then, $*$ is an $\Rc$-linear anti-involution of $\Ac$ such that $(C^{\lambda}_{S,T})^* = C^{\lambda}_{T,S}$.
\item[$\bullet$]
If $\lambda\in\Lambda$ and $S,T\in M(\lambda)$, then for any element $a\in\Ac$ we have 
\begin{align}\label{equ:aCLST}
  a\,C^{\lambda}_{S,T} \equiv \sum_{S'\in M(\lambda)}r_{a}(S',S)C^{\lambda}_{S',T}\quad 
   (\mathrm{mod}\;\Ac(<\lambda)),
\end{align}
where $r_a(S',S)\in\Rc$ is independent of $T$ and where $\Ac(<\lambda)$ is the $\Rc$-submodule of $\Ac$ 
generated by $\{C^{\mu}_{S'',T''}\,|\,\mu<\lambda;\,S'',T''\in M(\mu)\}$. 
\end{enumerate}
\end{defn}

\noindent
For each $\lambda\in\Lambda$, we define
\be
  \Ac(\{\lambda\}) : = \big\langle C_{S,T}^{\lambda}\,|\,S,T\in M(\lambda)\big\rangle_\Rc
\ee
and note that $\Ac\cong\bigoplus_{\lambda\in\Lambda}\Ac(\{\lambda\})$ as $\Rc$-modules.
For every $S_1,S_2,T_1,T_2\in M(\lambda)$, we also have
\be
 C^{\lambda}_{S_1,T_1}C^{\lambda}_{S_2,T_2} \equiv \psi(T_1,S_2)C^{\lambda}_{S_1,T_2} 
  \quad(\mathrm{mod}\;\Ac(<\lambda)),
\ee
where $\psi(T_1,S_2)\in\Rc$. This readily extends to a bilinear map 
$\psi\colon M(\lambda)\times M(\lambda)\to\Rc$.

Let the left, respectively right, $\Ac$-modules $W(\lambda)$ and $W(\lambda)^*$
be defined as the free $\Rc$-modules with basis $\{C_S\,|\,S\in M(\lambda)\}$ 
and $\Ac$-actions
\begin{align}\label{equ:aCsAction}
        aC_S: = \sum_{S'\in M(\lambda)}r_a(S',S)C_{S'},\qquad 
       C_Sa: = \sum_{S'\in M(\lambda)}r_{a^*}(S',S)C_{S'},\qquad 
         a\in\Ac.
\end{align}
We have the natural $\Rc$-module isomorphism
\begin{align}\label{equ:CLambdaWWIso}
  C^{\lambda}\colon  W(\lambda)\otimes_\Rc W(\lambda)^*\to\Ac(\{\lambda\}),\qquad 
  C_S\otimes C_T\mapsto C^{\lambda}_{S,T},
\end{align}
and the bilinear form
\be
  \phi_{\lambda}\colon W(\lambda)\times W(\lambda)\to\Rc,\qquad 
   (C_S,C_T)\mapsto \psi(S,T).
\ee
Letting $u,v,w\in W(\lambda)$ and $a\in\Ac(\{\lambda\})$, we note that the form is {\em symmetric}, 
$\phi_\lambda(u,v)=\phi_{\lambda}(v,u)$, {\em invariant} under the involution, 
$\phi_\lambda(a^*u,v) = \phi_{\lambda}(u,av)$, and satisfies $C^\lambda(u\otimes v)w=\phi_\lambda(v,w)u$.

The {\em radical} of $\phi_\lambda$ is defined as
\begin{align}
        \mathrm{rad}\,\phi_\lambda: = \{u\in W(\lambda)\,|\,\phi_\lambda(u,v)=0 \text{ for all }v\in W(\lambda)\}.
\end{align}
Let $\{C_{S_1},\ldots,C_{S_{d_{\lambda}}}\}$ be an
ordered basis for $W(\lambda)$, where $S_1,\ldots,S_{d_\lambda}\in M(\lambda)$.
Then, the symmetric matrix
\be
 G_\lambda: = 
 \begin{bmatrix}
 \phi_\lambda(C_{S_1},C_{S_1})&\phi_\lambda(C_{S_1},C_{S_2})&\ldots&
    \phi_\lambda(C_{S_1},C_{S_{d_\lambda}})\\
            \phi_{\lambda}(C_{S_2},C_{S_1}) & \phi_{\lambda}(C_{S_2},C_{S_2}) & \ldots & 
            \phi_\lambda(C_{S_2},C_{S_{d_{\lambda}}}) \vspace{0.1cm} \\
            \vdots  & \vdots  & \ddots & \vdots \\
            \phi_{\lambda}(C_{S_{d_{\lambda}}},C_{S_1}) & \phi_{\lambda}(C_{S_{d_{\lambda}}},C_{S_2}) & \ldots & 
            \phi_\lambda(C_{S_{d_{\lambda}}},C_{S_{d_{\lambda}}})\\
        \end{bmatrix}
\label{equ:CellGram}
\ee
is referred to as the corresponding {\em Gram matrix}. It holds that
\be
 \det G_\lambda\neq0
 \quad\Leftrightarrow\quad
 \mathrm{rad}\,\phi_\lambda=\{0\}
 \quad\Leftrightarrow\quad
 \{a\in\Ac(\{\lambda\})\,|\,av=0\text{ for all } v\in W(\lambda)\}=\{0\},
\ee
and if these relations are satisfied, the form $\phi_\lambda$ is said to be {\em non-degenerate}.
It follows that the representation
\be
 \rho: =\bigoplus_{\lambda\in\Lambda}\rho_\lambda\colon 
 \Ac\to\bigoplus_{\lambda\in\Lambda}\mathrm{End}_{\Rc}(W(\lambda))
\ee
is faithful if $\det G_{\lambda}\neq 0$ for all $\lambda\in\Lambda$.

\section{Temperley--Lieb algebras}
\label{Sec:TL}

Here, we work over $\Cbb$ and let the so-called loop fugacity, $\beta$, be an indeterminate unless otherwise 
stated. Most of the introductory material in Sections \ref{Sec:TLdef} and \ref{Sec:Cellularity} can be found in 
standard references on Temperley--Lieb algebras, 
such as \cite{Wenzl88,Martin91,GW93,Westbury95,Jones99,GL06,RSA14}.
New results include the decompositions in Proposition \ref{prop:spiders} and Corollary \ref{cor:TGamma} in 
Section \ref{Sec:TLT}, the free Baxterisation in Section \ref{Sec:TLYBI},
the classification of identity points and the determination of the corresponding principal 
hamiltonians in Section \ref{Sec:h}, the analysis of minimal hamiltonian polynomials in Section \ref{Sec:Min}, and 
the discussion of polynomial integrability in Section \ref{Sec:PolTL}.

\subsection{Definition}
\label{Sec:TLdef}

For each $n\in\Nbb_0$, let $B_n$ be the set of $2n$-tangles without input disks and loops. 
An element of $B_n$ is called a {\em connectivity diagram} and is thus a disk with $2n$ nodes connected 
by non-intersecting strings, defined up to regular isotopy. Examples are
\begin{align}
    \raisebox{-0.45cm}{\TLcircExI} \in B_2,\qquad
    \raisebox{-0.45cm}{\TLcircExII} \in B_3,\qquad
    \raisebox{-0.45cm}{\TLcircExIII} \in B_5.
\end{align}
Denote by $\mathrm{T}_{n}$ the complex vector space spanned by the elements of $B_n$. 
The dimension of $\mathrm{T}_{n}$ is given by a Catalan number,
\be
 \dim \mathrm{T}_n=\frac{1}{n+1}\binom{2n}{n}.
\ee

The \textit{Temperley--Lieb planar algebra} $\mathrm{TL}(\beta)$ is the graded vector space 
$(\mathrm{T}_{n})_{n\in\Nbb_0}$, together with the natural diagrammatic action of planar tangles, 
with each loop replaced by a factor of a parameter $\beta$.
For each $n\in\Nbb$, the \textit{Temperley--Lieb algebra} $\mathrm{TL}_{n}(\beta)$ is obtained by endowing 
the vector space $\mathrm{T}_{n}$ with the multiplication induced by the planar tangle $M_n$ 
in \eqref{equ:MultPlanar}. In accordance with (\ref{Aunit}), this is a unital algebra, 
and with reference to (\ref{Bnprime}), it holds that
\be
 a,b\in B_n'\quad\Longrightarrow\quad ab\in B_n'.
\label{abBn}
\ee

It is common to represent the connectivity diagrams in $B_n$ by rectangular diagrams, here drawn such that 
the marked boundary interval corresponds to the leftmost vertical side, as illustrated by
\be
 \raisebox{-0.45cm}{\CircToRecExC}\ \, \longleftrightarrow  \raisebox{-0.45cm}{\CircToRecExR}\ .
\ee
Accordingly, when describing connectivity diagrams in the disk picture, we refer to the boundary segment 
containing the first $n$ nodes, labelled clockwise (respectively counterclockwise) from the marked 
boundary component, as the {\em upper edge} (respectively the {\em lower edge}).

It is well known that the algebra $\mathrm{TL}_{n}(\beta)$ admits a presentation, 
$\langle e_1,\ldots,e_{n-1}\rangle$, where \cite{Kauffman87}
\be
    \idn\ \leftrightarrow \raisebox{-0.885cm}{\BMWRecIdOdd},\qquad 
    e_i\ \leftrightarrow\raisebox{-0.885cm}{\BMWRecEOdd}\qquad (i=1,\ldots,n-1),
\ee
subject to
\begin{align}
        e_i^2 &= \beta e_i,\\
        e_je_{i}e_j &= e_{j}, \quad\quad\quad |i-j|=1,\\
        e_ie_j &= e_je_i, \quad\quad\: |i-j|>1.
    \end{align}
\noindent
\textbf{Remark.} In physical applications, the parameter $\beta$ typically takes on a real (or complex) value 
and is often referred to as the {\em loop fugacity}.

\medskip

For each $n\in\Nbb$ and $\beta$ an indeterminate, 
there exists a unique $\wj_n\in\mathrm{TL}_{n}(\beta)$ such that
\be
 \wj_n^2=\wj_n,\qquad e_i\wj_n=\wj_ne_i=0,\qquad i=1,\ldots,n-1.
\label{wj}
\ee
We will have more to say about this {\em Jones--Wenzl idempotent} \cite{Wenzl87} in Section \ref{Sec:Min}.

\subsection{Standard modules and cellularity}
\label{Sec:Cellularity}

For each $n\in\Nbb$, let
\be
 \TLposet_n : = \{n-2k\,|\,k=0,\ldots,\lfloor\tfrac{n}{2}\rfloor\},
\ee
which is a naturally ordered set with $\min(D_n)=\frac{1}{2}(1-(-1)^n)$, and introduce the map
\begin{align}
        ^{\dagger}\colon B_{n} \to B_{n},\qquad
        x\mapsto x^\dagger,
\end{align}
where $x^\dagger$ is constructed by reflecting $x$ about the horizontal. 
Extended linearly to $\mathrm{T}_n$, this yields an anti-automorphism on $\mathrm{TL}_n(\beta)$.

For each $d\in \TLposet_{n}$, let $B_{n,d}\subseteq B_n$ denote the set of $n$-diagrams with exactly $d$
nodes on the lower edge connected to nodes on the upper edge. The $d$ strings connecting these 
$2d$ nodes are referred to as {\em through-lines}.
We now let $\TLset_{n,d}$ denote the set of all $n$-diagrams with $d$ through-lines, whose
upper edge has been discarded along with any string having both endpoints on the upper edge. 
The elements of $\TLset_{n,d}$ are referred to as {\em $(n,d)$-link states} and the elements of
\be
 L_n: =\bigcup_{d\in D_n}L_{n,d}
\label{LLnd}
\ee
as $n$-link states. To illustrate,
\be
 \mbox{}\qquad
 \raisebox{-0.45cm}{\TLLSExampleI}\ \mapsto\ \raisebox{-0.45cm}{\TLLSHalfExampleI},\qquad\quad
 \raisebox{-0.45cm}{\TLLSExampleII}\ \mapsto\ \raisebox{-0.45cm}{\TLLSHalfExampleI}
\label{62}
\ee
give rise to the same $(6,2)$-link state. In fact, the $(n,d)$-link states may be viewed as equivalence classes of 
$n$-diagrams. From that perspective, the two $6$-diagrams in (\ref{62}) are seen as representatives of the 
same $(6,2)$-link state.

The vector space
\be
 \TLmodule_{n,d}: =\mathrm{span}_{\Cbb}(\TLset_{n,d}),\qquad
 \dim\TLmodule_{n,d} = |\TLset_{n,d}|=\binom{n}{\frac{n-d}{2}} - \binom{n}{\frac{n-d}{2}-1},
\ee
becomes a $\mathrm{TL}_n(\beta)$-module by defining an action of the algebra generators on the link states
such that $a_2(a_1v)=(a_2a_1)v$ for all $a_1,a_2\in\mathrm{TL}_n(\beta)$ and $v\in\TLmodule_{n,d}$.
The action defining the familiar {\em standard module}, $\TLmodule_{n,d}$, is first given diagrammatically 
for $n$-diagrams acting on $(n,d)$-link states in the `natural way', see e.g.~\cite{RSA14},
and then extended linearly to all of $\mathrm{TL}_n(\beta)$ and all of $\TLmodule_{n,d}$.

For each pair $x,y\in L_{n,d}$, let $\langle x,y\rangle_{n,d}$ be constructed by reflecting the link state $x$ 
about the horizontal, placing it below the link state $y$, connecting the strands in the natural way, 
and replacing any loop by a factor of $\beta$, see e.g.~\cite{RSA14} for details. 
This extends to a bilinear map,
\be
        \langle\cdot,\cdot\rangle_{n,d}\colon \TLmodule_{n,d}\times\TLmodule_{n,d}\to\Cbb[\beta],\qquad
        (x,y)\mapsto \langle x,y\rangle_{n,d}.
\ee
Relative to the $(n,d)$-link state basis $L_{n,d}$, the nonzero elements of the corresponding 
Gram matrix $\TLgram_{n,d}$ are all monomials in $\beta$. The Gram determinant is thus polynomial in $\beta$, 
and following \cite{Westbury95}, it can be expressed as
\be
  \det\TLgram_{n,d} 
  =\prod_{j=1}^{\frac{n-d}{2}}
   \left(\frac{U_{d+j}(\frac{\beta}{2})}{U_{j-1}(\frac{\beta}{2})}\right)^{\!\dim\TLmodule_{n,d+2j}},
\label{equ:TLGramDet}
\ee
where $U_k(x)$ is the $k^{\mathrm{th}}$ Chebyshev polynomial of the section kind.

For each pair $x,y\in L_{n,d}$, let $|x\;y|_{n,d}$ be constructed by reflecting the link state $y$ about the horizontal, 
placing it above the link state $x$ and connecting the $d$ defects non-intersectingly. 
This extends to a bilinear map,
\be
  |\!\cdot\ \cdot|_{n,d}\colon \TLmodule_{n,d}\times\TLmodule_{n,d} \to \mathrm{TL}_{n,d}(\beta),
\label{equ:vsIso}
\ee
where $\mathrm{TL}_{n,d}(\beta)$ is the subset of $\mathrm{TL}_{n}(\beta)$ whose elements have exactly $d$ 
through-lines. It follows that
\begin{align}
    |x\;y|_{n,d}\,z= \langle y, z\rangle_{n,d}\,x,\qquad \forall\,x,y,z\in\TLmodule_{n,d}.
\end{align}
When clear, the subscripts of $\langle\cdot,\cdot\rangle_{n,d}$ and $|\!\cdot\ \cdot|_{n,d}$ will be suppressed, 
writing $\langle\cdot,\cdot\rangle$ and $|\!\cdot\ \cdot|$, respectively. 

It now follows \cite{GL06} that $\mathrm{TL}_n(\beta)$ is cellular with cell datum
$(\TLposet_n,L_n, |\!\cdot\ \cdot|, \dagger)$, where the dagger provides an adjoint operation 
relative to the bilinear form $\langle\cdot,\cdot\rangle$ on $\TLmodule_{n,d}$:
\begin{align}
        \langle x,ay\rangle = \langle a^{\dagger}x,y\rangle,\qquad a\in\mathrm{TL}_n(\beta),\ x,y\in\TLmodule_{n,d}.
\end{align}

In preparation for the discussion in Section \ref{Sec:Min}, let
\be
 V_n:=\mathrm{span}_{\Cbb}(L_n),
\ee
and note that, as vector spaces,
\be
 V_n=\bigoplus_{d\in D_n}V_{n,d},
\ee
hence
\be
 \dim V_n=\sum_{d\in D_n}\dim V_{n,d}=\binom{n}{\lfloor\frac{n}{2}\rfloor}.
\label{dimVn}
\ee
We also let $\rho_{n,d}$ 
denote the representation corresponding to the standard module $V_{n,d}$, and let
\be
 \rho_n:=\bigoplus_{d\in D_n}\rho_{n,d},
\ee
so relative to an ordered $V_n$-basis of the form
\be
 L_{n,s_n}\cup L_{n,s_n+2}\cup\cdots\cup L_{n,n},\qquad s_n:=\tfrac{1}{2}(1-(-1)^n),
\label{LLL}
\ee
the matrix representation of $\rho_n$ is block-diagonal.
Moreover, from Section \ref{Sec:Cellular}, we have that $\rho_n$ is faithful for all $\beta\in\Cbb$ for which
$\prod_{d\in D_n}\det G_{n,d}\neq0$. In particular, $\rho_n$ is faithful for $\beta$ an indeterminate
and fails to be faithful for at most finitely many values $\beta\in\Cbb$.

\subsection{Transfer operators}
\label{Sec:TLT}

The connectivity bases for $\mathrm{T}_2$ and $\mathrm{T}_1$ are given by $B_2=\{\mathds{1}_2, e\}$ and 
$B_1=\{\mathds{1}_1\}$, respectively. Imposing that the coefficients to the unit elements are nonzero in 
(\ref{RKK}), the $R$- and $K$-elements can be normalised such that
\be
 R(u)=\mathds{1}_2+ue,\qquad K(u)=\Ko(u)=\mathds{1}_1,
\label{RuKu}
\ee
where $u\in\Cbb$. Diagrammatically, this $R$-element is given by
\be
    R(u)=\raisebox{-0.325cm}{\RotRopGreenOpBBTLST}\, 
    =  \raisebox{-0.325cm}{\RotRopTLxIdST} 
    +  u\raisebox{-0.325cm}{\RotRopTLxEST}\ ,
\label{RTL}
\ee
and the transfer operator morphs into its familiar (see e.g.~\cite{PRZ06}) diagrammatic form,
\begin{align}\label{equ:TnBaxGen}
    T_n(u)= \raisebox{-0.925cm}{\DRTMnNode}\; .
\end{align}
To emphasise its dependence on $\beta$, we occasionally write $T_n(u,\beta)$ instead.
\medskip

\noindent
{\bf Remark.}
The parameterisation in (\ref{RuKu}) ensures that $R(u)\neq0$ for all $u$, and this would similarly have been
achieved had we chosen to work with $\hat{R}(\hat{u})=\hat{u}\,\mathds{1}_2+e$, where $\hat{R}(\hat{u})\neq0$ 
for all $\hat{u}\in\Cbb$. Viewing $u$ and $\hat{u}$ as coordinates on the Riemann sphere, we see that, for all 
$u\in\Cbb\cup\{\infty\}$, $R(u)\mapsto R(\hat{u})=\frac{1}{u}\hat{R}(u)$ as $u\mapsto\hat{u}\equiv\frac{1}{u}$.
\medskip

With (\ref{RTL}), we have the decompositions
\be
  \raisebox{-0.95cm}{\LeftCloseTLtransfer} 
  = \raisebox{-0.95cm}{\LeftCloseTLtransferId} 
  + (u+v+\beta uv)\raisebox{-0.95cm}{\LeftCloseTLtransferE}\ \ , \qquad\quad
  \raisebox{-0.95cm}{\RightCloseTLtransfer} 
  = uv\raisebox{-0.95cm}{\RightCloseTLtransferId} 
  + (u+v+\beta)\raisebox{-0.95cm}{\RightCloseTLtransferE}\ \ .
\label{uv}
\ee
They reduce to the so-called ``drop-down" relations or properties \cite{PRV10} for $u+v+\beta uv=0$ and 
$u+v+\beta=0$, respectively, see also Proposition \ref{prop:HamLimitPointsGen}.

\begin{prop}\label{prop:special}
We have
\be
 T_n(u)\big|_{u(2+\beta u)=0}=\beta\,\idn,\qquad
 T_n(u)\big|_{\beta+2u=0}=\beta u^{2n}\idn.
\ee
\end{prop}
\begin{proof}
The results follow from repeated application of (\ref{uv}) with $v=u$, to (\ref{equ:TnBaxGen}).
\end{proof}
\begin{prop}\label{prop:Tnc}
We have
\be
 T_n(u)= \raisebox{-0.925cm}{\DRTMnNodec}\; .
\label{Tswapped}
\ee
\end{prop}
\begin{proof}
Together with Proposition \ref{prop:special}, relations similar to the ones in (\ref{uv}) applied to the righthand 
side of (\ref{Tswapped}) establish the result for $u$ such that $u(2+\beta u)(\beta+2u)=0$. 
If $u(2+\beta u)(\beta+2u)\neq0$, then the operators
\be
  \raisebox{-0.325cm}{\RotRopRedOpBBiaTLcross}\;
 = \frac{u^2-1}{2u+\beta}\raisebox{-0.325cm}{\RotRopTLxIdST} +\raisebox{-0.325cm}{\RotRopTLxEST}\ ,\qquad
 \raisebox{-0.325cm}{\RotRopDarkRedOpBBiaTLcross}\;
 =\frac{1-u^2}{u(2+\beta u)}\raisebox{-0.325cm}{\RotRopTLxIdST} +\raisebox{-0.325cm}{\RotRopTLxEST}\;,
\ee
satisfy
\be
 \raisebox{-0.45cm}{\RotMarkRedInvrFilterediTLcross}\ 
 =\; \raisebox{-0.45cm}{\MarkRedIdFilteredInv}\ ,\qquad\quad 
 \raisebox{-0.9cm}{\RotMarkRedYBEgoRightSingleRiaTLcross}\ 
 = \raisebox{-0.9cm}{\RotMarkRedYBEgoLeftSingleRiaTLcross}
\ee
and
\be
 \raisebox{-0.45cm}{\LeftCapTLtransfer} 
 =\frac{u(2+\beta u)}{\beta+2u}\raisebox{-0.45cm}{\LeftCapTLtransferE}\ , \qquad\quad
 \raisebox{-0.45cm}{\RightCapTLtransfer}\,
 = \frac{\beta+2u}{u(2+\beta u)}\raisebox{-0.45cm}{\RightCapTLtransferE}\ .
\ee
It follows that
\begin{align}
 T_n(u)=\!\raisebox{-0.925cm}{\TLtransferCrossInv} 
 \,=\!\raisebox{-0.925cm}{\TLtransferCrossInvI}
 \,=\!\raisebox{-0.925cm}{\DRTMnNodec}.
\end{align}
\end{proof}
For $u\neq0$, the $R$-operator is crossing symmetric,
\be\label{equ:RCrossingSym}
 \raisebox{-0.325cm}{\RotRopGreenOpBBTLSTRotu}\;=\, u \raisebox{-0.325cm}{\RotRopGreenOpBBTLSTuinv}\ ,
\ee
with $u=1$ an isotropic point.
By Proposition \ref{prop:Tnc},
\be
 T_n(u) = u^{2n}T_n(1/u),
\label{Tcross}
\ee
so $u^{-n}T_n(u)$ is invariant under $u\mapsto1/u$. Since $T_n(u)$ is polynomial in $u$ of degree at most $2n$, 
it follows that there exists $\Tt_n(x)\in\mathrm{TL}_n(\beta)[x]$ such that
\be
 T_n(u)=u^n\Tt_n(u+\tfrac{1}{u}).
\label{Tt}
\ee
Moreover, there exist $a_0,\ldots,a_{2n}\in\mathrm{TL}_n(\beta)$ such that
\be
 T_n(u)=\sum_{i=0}^{2n}a_iu^i,
\ee
and using (\ref{Tcross}), it follows that $a_{2n-i}=a_i$, $i=0,\ldots,n-1$, so
\be
 \Tt_n(x)=a_n+2\sum_{i=1}^{n}a_{n-i}T_{i}^{(c)}\big(\tfrac{x}{2}\big),
\label{Tt2}
\ee
where $T_k^{(c)}$ is the $k^{\mathrm{th}}$ Chebyshev polynomial of the first kind.
In establishing (\ref{Tt2}), we have used the familiar relation
\be
 T_k^{(c)}(\cosh\theta)=\cosh(k\theta),\qquad \theta\in\Cbb.
\ee

To shine further light on the structure of $T_n(u,\beta)$, 
we introduce the following parameterised elements of $\mathrm{TL}_n(\beta)$: 
For $n\in\Nbb$ and each pair $j,k\in\Nbb_0$ such that $j+k\leq n-2$, let
\begin{align}
    S^{(n)}_{j,k}(u) := \raisebox{-1.85cm}{\TransferSpiderSmaller},
\label{spider}
\end{align}
which reduces to $S^{(n)}_{j,n-j-2}(u)=e_{j+1}$ for $k=n-j-2$. To emphasise its dependence on $\beta$, 
we occasionally write $S^{(n)}_{j,k}(u,\beta)$. Note that $S^{(n)}_{j,k}(u)\in \mathrm{span}_{\Cbb}(B'_n)$.
\begin{lem}\label{lem:squid}
For $n\in\Nbb$, we have
\be
 \raisebox{-1.82cm}{\TransferSquidNew}=u^{2n-2}\idn+(\beta+2u)\sum_{k=0}^{n-2}u^{2k}S^{(n)}_{0,k}(u)
\label{squid}
\ee
and
\be
 \raisebox{-1.82cm}{\leftSquidN}=\idn+u(2+\beta u)\sum_{j=0}^{n-2}S_{j,0}^{(n)}(u).
\label{squid2}
\ee
\end{lem}
\begin{proof}
The result (\ref{squid}) follows by induction on $n$, decomposing the two rightmost $R$-operators as in (\ref{uv})
and using that both sides of (\ref{squid}) reduce to $\mathds{1}_1$ for $n=1$.
The result (\ref{squid2}) follows similarly.
\end{proof}

\noindent
For $k\in\Nbb_0$ and $x\in\Cbb$, we let
\be
 [k]_x\!:=1+x+\cdots+x^{k-1}\quad (k>1),\qquad [1]_x\!:=1,\qquad[0]_x\!:=0.
\ee
\begin{prop}\label{prop:spiders}
The transfer operator decomposes as
\be
   T_n(u,\beta) = \big(\beta[n+1]_{u^2} + 2u[n]_{u^2}\big)\idn 
     + u(\beta + 2u)(2+\beta u)\sum_{j=0}^{n-2}\sum_{k=0}^{n-2-j} \!\!u^{2k}S^{(n)}_{j,k}(u).
\label{equ:DTMHpolyRefinedSC}
\ee
\end{prop}
\begin{proof}
By (\ref{uv}) with $v=u$, we have
\be
 T_n(u)=u(2+\beta u)  \raisebox{-0.91cm}{\TransferSquidNewNoBrace}
 + \raisebox{-0.91cm}{\IdTnto}\ .
\label{Tdecomp}
\ee
The result now follows by induction on $n$, applying (\ref{squid}) to the first term on the right 
in (\ref{Tdecomp}) and the induction hypothesis to the second term.
\end{proof}
\begin{cor}\label{cor:spidercrossing}
For $u(\beta+2u)(2+\beta u)\neq0$, we have
\be
 \frac{1}{u^{n-2}}\sum_{j=0}^{n-2}\sum_{k=0}^{n-2-j} \!\!u^{2k}S^{(n)}_{j,k}(u)
 =u^{n-2}\sum_{j=0}^{n-2}\sum_{k=0}^{n-2-j} \!\!\tfrac{1}{u^{2k}}S^{(n)}_{j,k}\big(\tfrac{1}{u}\big).
\ee
\end{cor}
\begin{proof}
The result follows from (\ref{Tcross}) and Proposition \ref{prop:spiders}.
\end{proof}
\begin{cor}\label{cor:TGamma}
The transfer operator decomposes uniquely as
\be
  T_n(u,\beta) = \big(\beta[n+1]_{u^2} + 2u[n]_{u^2}\big)\idn 
   + u(\beta + 2u)(2+\beta u)\sum_{a\in B_n'}\Gamma^{(n)}_a(u,\beta) a,
\ee
where $\Gamma^{(n)}_a(u,\beta)$ is polynomial in $u,\beta$ for every $a\in B_n'$.
\end{cor}
\begin{proof}
With the parameterisation (\ref{RuKu}), the decomposition of $T_n(u,\beta)$ into connectivity diagrams 
(elements of $B_n$) involves only coefficients that are polynomial in $u,\beta$, 
and since $B_n$ is a linearly independent set, the decomposition is unique.
The restriction to a summation over $B'_n$ is permitted (and required for uniqueness)
because $S^{(n)}_{j,k}(u)\in \mathrm{span}_{\Cbb}(B'_n)$.
\end{proof}

\noindent
\textbf{Remark.} The expression (\ref{equ:TnBaxGen}) for $T_n(u,\beta)$ may be formally extended to $n=0$, 
yielding $T_0(u,\beta)=\beta\mathds{1}_0$.
With $\mathds{1}_0\equiv1$, this becomes $T_0(u,\beta)=\beta$.
\begin{prop}\label{prop:TnTna}
We have
\be
  [T_n(u,\beta),T_n(v,\beta)] = \sum_{a\in B_n'}p_a(u,v,\beta) a,
\label{TnuTnv}
\ee
where $p_a(u,v,\beta)$ is polynomial in $u,v,\beta$ for every $a\in B_n'$.
\end{prop}
\begin{proof}
The result follows from Corollary \ref{cor:TGamma}, noting that the restriction to a summation over $B'_n$ 
is permitted because of (\ref{abBn}).
\end{proof}

Let (indeterminate or nonzero scalar) $q$ be such that
\be
 \beta=q+q^{-1},
\label{q}
\ee
and define
\be
 F_n:=q^{-n}T_n(-q),\qquad \Fo_n:=q^nT_n(-q^{-1}).
\label{F}
\ee
Although these {\em braid elements} are known, see e.g.~\cite{MDSA11}, to be central and satisfy $F_n=\Fo_n$, 
we now present a proof of this in our notation, for completeness.
\begin{prop}
For every $n\in\Nbb$, we have
\be
 (\mathrm{i})\ \ \Fo_n=F_n,\qquad (\mathrm{ii})\ \ F_n\in\Zc\big(\mathrm{TL}_n(\beta)\big).
\ee
\end{prop}
\begin{proof}
Property (i) follows from Corollary \ref{cor:spidercrossing}.
As to (ii), we have
\be
 \raisebox{-0.95cm}{\TLDDtop} = q^2\raisebox{-0.95cm}{\TLDDtopId}\;,\qquad\quad
  \raisebox{-1.45cm}{\TLDDbot} = q^2 \raisebox{1.1cm}{\scalebox{1}[-1]{\TLDDtopId}}\;.
\ee
It follows that
\be
 e_iF_n
 =q^2\raisebox{-1.3cm}{\BraidTLComExpr}
 =F_ne_i,\qquad i=1,\ldots,n-1,
\ee
hence $F_n\in\Zc\big(\mathrm{TL}_n(\beta)\big)$.
\end{proof}
\noindent
We note that
\be
 F_n=\big(q^{n+1}+q^{-n-1}\big)\idn
   -\frac{(q^2-1)^2}{q^n}\sum_{j=0}^{n-2}\sum_{k=0}^{n-2-j} \!\!q^{2k}S^{(n)}_{j,k}(-q,q+q^{-1}).
\ee

\subsection{Yang--Baxter integrability}
\label{Sec:TLYBI}

The parameterisations (\ref{RuKu}) provide a Baxterisation for all $\beta$, as
\be
 \raisebox{-0.325cm}{\RotRopDarkRedOpBBi}\;
 =\yo^{(i)}(u,v)\raisebox{-0.325cm}{\RotRopTLxIdST} +\raisebox{-0.325cm}{\RotRopTLxEST}\; ,\qquad
 \raisebox{-0.325cm}{\RotRopRedOpBBi}\;
 =y^{(i)}(u,v)\raisebox{-0.325cm}{\RotRopTLxIdST} +\raisebox{-0.325cm}{\RotRopTLxEST}\; ,
\label{tlaux}
\ee
where
\be
 \yo^{(1)}(u,v)=\frac{1-uv}{u+v+\beta uv},\qquad y^{(1)}(u,v)=\frac{uv-1}{u+v+\beta},
\label{y1}
\ee
and
\begin{align}
 \yo^{(2)}(u,v)=\yo^{(3)}(u,v)=\frac{u-v}{1+\beta v+uv},\qquad
 y^{(2)}(u,v)=y^{(3)}(u,v)=\frac{v-u}{1+\beta u+uv},
\label{y2y3}
\end{align}
solve the inversion relations, YBEs and BYBEs in Proposition \ref{Prop:YBE}. It follows 
that $T_n(v)$ commutes with $T_n(u)$ (that is, $T_n(v)\in\Cc_{\mathrm{TL}_n(\beta)}(T_n(u))$) for all $v$
for which the functions $\yo^{(i)}(u,v)$ and $y^{(i)}(u,v)$, $i=1,2,3$, are well-defined.
The next result extends this commutativity to all $u,v,\beta\in\mathbb{C}$.
\begin{prop}\label{Prop:TT0}
For every $\beta\in\Cbb$, we have
\be
 [T_n(u,\beta),T_n(v,\beta)]=0,\qquad\forall\,u,v\in\Cbb.
\label{TnTn0}
\ee
\end{prop}
\begin{proof}
For each $\beta$, Proposition \ref{Prop:YBE} implies that the commutator is zero for all $u,v\in\Cbb$ 
except possibly along the algebraic curves
\be
 u+v+\beta uv=0,\qquad u+v+\beta=0,\qquad 1+\beta v+uv=0,\qquad 1+\beta u+uv=0,
\ee
defined by setting the denominators in (\ref{y1}) and (\ref{y2y3}) to zero.
That is, for every $(u,v)\in\Cbb^2$ not on any of the curves, the polynomials in Proposition \ref{prop:TnTna}
satisfy $p_a(u,v,\beta)=0$. It follows that, for all $a\in B_n'$, $p_a(u,v,\beta)=0$ for all $(u,v)\in\Cbb^2$, 
hence $[T_n(u,\beta),T_n(v,\beta)]=0$ for all $u,v\in\Cbb$.
\end{proof}

We note that the auxiliary operators (\ref{tlaux}) can be expressed in terms of the $R$-operator as
\be
 \raisebox{-0.325cm}{\RotRopDarkRedOpBBi}\;=\yo^{(i)}(u,v)R\big(\tfrac{1}{\yo^{(i)}(u,v)}\big),\qquad
 \raisebox{-0.325cm}{\RotRopRedOpBBi}\;=y^{(i)}(u,v)R\big(\tfrac{1}{y^{(i)}(u,v)}\big).
\ee
Accordingly, for the Temperley--Lieb planar algebra, the full generality of the sufficient 
conditions \eqref{Invs}--\eqref{BYBEs} is not exploited. To contrast, we refer to \cite{PR23} for an example 
of an integrable planar-algebraic model relying nontrivially on the generality of the sufficient conditions.

We stress that the integrability (\ref{TnTn0}) is present {\em without} 
constraining the parameterisation of the $R$-operator in (\ref{RuKu}). 
We accordingly refer to the Baxterisation as {\em free} and the Temperley--Lieb loop model as 
{\em freely Baxterisable}.
The usual (here renormalised) trigonometric Baxterisation \cite{PRZ06},
\be
R(u)=\!\raisebox{-0.325cm}{\RotRopGreenOpBBTLST}\,
 =y(u)\raisebox{-0.325cm}{\RotRopTLxIdST} +\raisebox{-0.325cm}{\RotRopTLxEST}\; ,\qquad
 y(u)=\frac{\sin(\lambda-u)}{\sin u},\qquad
 \beta=2\cos\lambda,
\ee
is not free, as
\be
 y^2(u)+\beta y(u)+1=\frac{\sin^2\!\lambda}{\sin^2\!u}.
\ee
It nevertheless provides a solution to the relations in Proposition \ref{Prop:YBE},
as we can normalise the auxiliary operators such that
\begin{align}
    \raisebox{-0.325cm}{\RotRopDarkRedOpBBia}\;&=R(u+v), &
    \raisebox{-0.325cm}{\RotRopRedOpBBia}\;&=R(2\lambda-u-v),
 \\[.15cm]
    \raisebox{-0.325cm}{\RotRopDarkRedOpBBib}\;=\raisebox{-0.325cm}{\RotRopDarkRedOpBBic}\;
    &=R(\lambda-u+v), &
\hspace{-1.3cm}
    \raisebox{-0.325cm}{\RotRopRedOpBBib}\;=\raisebox{-0.325cm}{\RotRopRedOpBBic}\;&=R(\lambda+u-v).
\end{align}
In this case, the YBEs \eqref{YBE} assume the standard \textit{additive} form. 
Crossing symmetry of the $R$-operator and the (trigonometric) transfer operator amounts to
\be
     \raisebox{-0.325cm}{\RotRopGreenOpBBTLST} 
     \ = y(u)\raisebox{-0.325cm}{\RotRopGreenOpBBTLSTRotlau}\;,\qquad
     y^{-n}(u)T_n(u) =y^{-n}(\tilde{u})T_n(\tilde{u}),\qquad
  \tilde{u} \equiv\lambda - u.
\ee

\subsection{Polynomial integrability}
\label{Sec:TLB}

Having established the (Yang--Baxter) integrability of the Temperley--Lieb loop model in Proposition \ref{Prop:TT0}, 
we now turn to the integrals of motion of the model. 
It is known, see for example \cite{MDRRSA16}, that the bilinear form $\langle\cdot,\cdot\rangle_{n,d}$,
with the corresponding scalar field restricted to $\Rbb$,
is an inner product for large enough $\beta\in\Rbb$.
In fact, it follows from (\ref{equ:TLGramDet}) that, for $\beta>2\cos\frac{\pi}{n}$, 
the form is an inner product and $\mathrm{TL}_n(\beta)$ is semisimple.
Since the transfer operator is self-adjoint, $T_n^{\dagger}(u,\beta) = T_n(u, \beta)$, 
with respect to the form, it is thus diagonalisable for $\beta>2\cos\frac{\pi}{n}$ and $u\in\Rbb$. 
Using this, the following result follows from Corollary \ref{Cor:diagsemi}.
\begin{prop}\label{prop:TLIntSuff}
  For $\beta>2\cos\frac{\pi}{n}$, $n\in\Nbb$, and $u\in\Rbb$, 
  the Temperley--Lieb loop model described by $T_n(u, \beta)$ is polynomially integrable.
\end{prop}
\noindent
\textbf{Remark.}
For $\beta>2$, the Temperley--Lieb planar algebra is, in fact, a subfactor planar algebra.
\medskip

\noindent
Proposition \ref{prop:TLIntSuff} implies that there exists a $u$-independent integral of motion 
$b_n\in\mathrm{TL}_n(\beta)$ such that $T_n(u, \beta)\in\Rbb[b_n]$, and following Section \ref{Sec:Ass}, 
$\rho_n(b_n)$ and $\rho_n(T_n(u, \beta))$ will have closely related Jordan decompositions.
We note that the usual hamiltonian (see (\ref{h0}) below) has this property when $T_n(u, \beta)$ is expanded to 
lowest nontrivial order in $u$. In the following, we derive the principal hamiltonians associated with 
$T_n(u, \beta)$, see (\ref{h0}) and (\ref{hustar}), and use spectral analysis to argue that both of these 
$\mathrm{TL}_n(\beta)$-elements can indeed play the role of $b_n$, at least for small $n$. We supplement this 
result by determining an explicit polynomial expression for $T_n(u, \beta)$ in terms of each of the hamiltonian 
elements, and find that they are well-defined for all $u\in\Cbb$ and all but finitely many $\beta$-values in $\Cbb$. 
The restrictions on $u$ and $\beta$ in Proposition \ref{prop:TLIntSuff} can thus be relaxed accordingly, at least for 
small $n$. Moreover, we find that, for small $n$, $T_n(u, \beta)$ is polynomial in at least one of the two principal 
hamiltonians for all $\beta,u\in\Cbb$, see the discussion following (\ref{EnEnempty}).
\medskip

\noindent
\textbf{Remark.}
In addition to {\em establishing} polynomial integrability, our focus in the following is on probing the naturally arising 
principal hamiltonians as candidates for the integral of motion $b_n$. In fact, one could also explore whether any 
given specialisation of $T_n(u)$, where $u$ is fixed to some value, could play the role of $b_n$. We have indeed 
examined a number of such candidates, including the
braid-limit element $F_n$ in (\ref{F}), but the standard hamiltonian (\ref{h0}) has so far had the fewest number 
of {\em exceptional points}, see the discussion in Section \ref{Sec:PolTL}.

\subsection{Hamiltonian limits}
\label{Sec:h}

Proposition \ref{prop:special} gives sufficient conditions for the determination of identity points.
We now classify the identity points for $T_n(u,\beta)$, $n\geq2$. 
\begin{prop}\label{prop:TLIdPts}
Let $n\in\Zbb_{\ge 2}$. For $\beta\notin\{-2,0,2\}$, the set of identity points for $T_n(u,\beta)$ is given by 
$\{0,-\frac{\beta}{2},-\frac{2}{\beta}\}$. For $T_n(u,\pm2)$, the set of identity points is given by $\{0,\mp1\}$.
For $T_n(u,0)$, the only identity point is $u_*=0$.
\end{prop}
\begin{proof}
It follows from (\ref{spider}) that the connectivity diagram corresponding to $e_1\cdots e_{n-1}$ only appears in 
$S_{j,k}^{(n)}$ for $j=k=0$, and with coefficient $u^{n-2}$. By Proposition \ref{prop:spiders}, the element thus 
appears in $T_n(u,\beta)$ with coefficient $u^{n-1}(\beta+2u)(2+\beta u)$.
This expression vanishes exactly for the indicated values for $u$.
\end{proof}

To determine the hamiltonian associated with the identity point $u_*=0$, 
we use Proposition \ref{prop:spiders} to compute
\begin{align}
 T_n(\epsilon,\beta)\big|_{\beta\neq0}
 &=(\beta+2\epsilon)\idn+2\epsilon\beta\sum_{j=1}^{n-1}e_j+\Oc(\epsilon^2),
\label{Tneps}
 \\
 \tfrac{1}{2\epsilon}T_n(\epsilon,0)&=\idn+2\epsilon\sum_{j=1}^{n-1}e_j+\Oc(\epsilon^2).
\end{align}
For $n\geq2$ and all $\beta$, we may thus choose the familiar (see e.g.~\cite{PS90,PRZ06,PR07})
\be
 h_0:=-\sum_{i=1}^{n-1}e_i
\label{h0}
\ee
as the principal hamiltonian associated with $u_*=0$.
\medskip

\noindent
{\bf Remark.}
There is also a `hidden' identity point at infinity, see the Remark, following (\ref{equ:TnBaxGen}), that addresses
the extension of the domain for $u$ from $\Cbb$ to the Riemann sphere. The corresponding principal hamiltonian
is proportional to $h_0$.
\medskip

Hamiltonians associated with the identity points $u_*=-\frac{\beta}{2}\neq0$ and $u_*=-\frac{2}{\beta}$ 
do not seem to have been discussed before in the literature.
To determine the corresponding principal hamiltonians, $h_{\text{-}\frac{\beta}{2}}$ and $h_{\text{-}\frac{2}{\beta}}$, 
we expand as
\begin{align}
 T_n(-\tfrac{\beta}{2}+\epsilon,\beta)\big|_{\beta\neq0,\pm2}
 &=\Big(\tfrac{\beta^{2n+1}}{4^n}+\epsilon\big(2[n]_{\frac{\beta^2}{4}}-\tfrac{n\beta^{2n}}{4^{n-1}}\big)\!\Big)\idn
  +\epsilon(\beta^2-4)\sum_{j=0}^{n-2}\sum_{k=0}^{n-2-j} \!\!\big(\tfrac{\beta}{2}\big)^{2k+1}
     S^{(n)}_{j,k}(-\tfrac{\beta}{2})+\Oc(\epsilon^2),
 \label{Tnbeta}
 \\
 T_n(-\tfrac{2}{\beta}+\epsilon,\beta)\big|_{\beta\neq0,\pm2}
 &=\big(\beta
  -2\epsilon[n]_{\frac{4}{\beta^2}}\big)\idn
  -\epsilon(\beta^2-4)\sum_{j=0}^{n-2}
    \sum_{k=0}^{n-2-j} \!\!\big(\tfrac{2}{\beta}\big)^{2k+1}S^{(n)}_{j,k}(-\tfrac{2}{\beta})+\Oc(\epsilon^2),
 \end{align}
and
\begin{align}
 T_n(-1+\epsilon,2)&=2(1-\epsilon n+\epsilon^2n^2)\idn
   -4\epsilon^2\sum_{j=0}^{n-2}\sum_{k=0}^{n-2-j} \!\!S^{(n)}_{j,k}(-1,2)+\Oc(\epsilon^3),
 \\[.1cm]
 T_n(1+\epsilon,-2)&=-2(1+\epsilon n+\epsilon^2n^2)\idn
   -4\epsilon^2\sum_{j=0}^{n-2}\sum_{k=0}^{n-2-j} \!\!S^{(n)}_{j,k}(1,-2)+\Oc(\epsilon^3).
\label{Tnm2}
\end{align}
\begin{prop}
For $n\geq2$, $\beta\neq0$, and up to rescaling, 
the principal hamiltonian for $u_*\in\{-\frac{\beta}{2},-\frac{2}{\beta}\}$ is given by
\be
 h_{u_*}=\frac{1}{u_*^{n-2}}\sum_{j=0}^{n-2}\sum_{k=0}^{n-2-j} \!u_*^{2k}S^{(n)}_{j,k}(u_*).
\label{hustar}
\ee
With the chosen normalisation, it holds that $h_{\text{-}\frac{\beta}{2}}=h_{\text{-}\frac{2}{\beta}}$.
\end{prop}
\begin{proof}
The expression (\ref{hustar}) follows from (\ref{hprincipal}) and the expansions (\ref{Tnbeta})--(\ref{Tnm2}).
The relation $h_{\text{-}\frac{\beta}{2}}=h_{\text{-}\frac{2}{\beta}}$ follows from Corollary \ref{cor:spidercrossing}.
\end{proof}
\noindent
Although $h_{n,\text{-}\frac{2}{\beta}}$ and $h_{n,\text{-}\frac{\beta}{2}}$ are linearly dependent,
$h_{n,\text{-}\frac{2}{\beta}}$ and $h_{n,0}$ are not.
We also note that
\be
 h_{u_*}\in\mathrm{span}_{\Cbb[\beta]}(B_n'),\qquad u_*\in\{0,-\tfrac{2}{\beta}\}.
\label{hB}
\ee
For $n=2,3,4,5$, the principal hamiltonian $h_{n,\text{-}\frac{2}{\beta}}$ is given in Appendix \ref{Sec:Apphbeta}.

At the isotropic point $u=1$, Proposition \ref{prop:spiders} implies that
\be
 T_n(1,\beta)=\big((n+1)\beta+ 2n\big)\idn
  +(\beta+2)^2\sum_{j=0}^{n-2}\sum_{k=0}^{n-2-j} \!\!S^{(n)}_{j,k}(1,\beta),
\ee
and we note that $T_n(1,-2)=-2\idn$, in accordance with (\ref{Tnm2}).

\subsection{Minimal hamiltonian polynomials}
\label{Sec:Min}

Since $\mathrm{TL}_{n}(\beta)$ is finite-dimensional, corresponding to each $a\in\mathrm{TL}_{n}(\beta)$,
there exists a unique monic polynomial, of least positive degree, that annihilates $a$ -- the so-called
{\em minimal polynomial} of $a$.
Let $m_{u_*}^{(n)}$ denote the minimal polynomial of $h_{n,u_*}$ for $\beta$ an indeterminate, 
and let $m_{u_*,\beta}^{(n)}$ denote the minimal polynomial of $h_{n,u_*}$ for $\beta\in\Cbb$. 
Examples are provided in Appendix \ref{Sec:Apph0} and Appendix \ref{Sec:Apphbeta}.
We denote the degrees of the {\em minimal hamiltonian polynomials},
$m_{u_*}^{(n)}$ and $m_{u_*,\beta}^{(n)}$, by
\be
 l_{u_*}^{(n)}:=\deg\!\big(m_{u_*}^{(n)}\big),\qquad
 l_{u_*,\beta}^{(n)}:=\deg\!\big(m_{u_*,\beta}^{(n)}\big).
\ee
For ease of presentation, we let
\be
 c_n:=\binom{n}{\lfloor\frac{n}{2}\rfloor}.
\ee
\begin{prop}\label{prop:llbin}
For each $n\in\Zbb_{\geq2}$ and $u_*\in\{0,-\frac{2}{\beta}\}$, we have
\be
 l_{u_*,\beta}^{(n)}\leq l_{u_*}^{(n)}\leq c_n.
\label{ineq}
\ee
\end{prop}
\begin{proof}
For $\rho_n$ faithful, the minimal polynomial of $\rho_n(h_{u_*})$ is the same as that of $h_{u_*}$,
irrespective of $\beta$ being an indeterminate or taking on a complex value.
Specialising $\beta$ to a complex value may introduce 
spurious (see Section \ref{Sec:Spectral} and the remark following (\ref{lbin})) 
degeneracies in the spectrum of $\rho_n(h_{u_*})$, 
and such degeneracies could reduce the degree of the minimal polynomial of $\rho_n(h_{u_*})$.
This explains the first inequality. The second inequality follows from the existence of
an $c_n$-dimensional representation, $\rho_n$, that is faithful for $\beta$ an indeterminate.
\end{proof}
\noindent
\textbf{Remark.}
To appreciate the inequality $l_{u_*,\beta}^{(n)}\leq c_n$ directly, 
note that $\rho_n$ is a $c_n$-dimensional representation
that is faithful for all but finitely many $\beta$-values. The degree of the minimal polynomial 
for $\beta$ complex and generic is thus bounded by $c_n$,
and, possibly rescaled to remain well-defined, the corresponding minimal polynomial will remain annihilating
when specialising $\beta$ to one of these values. Such a rescaling can be chosen such that the rescaled 
polynomial is nonzero
when specialising $\beta$, and the degree of this rescaled polynomial may decrease upon specialisation 
(this happens if and only if the rescaling multiplies the leading monomial by a factor that is zero when specialised)
but cannot increase.
\begin{prop}\label{prop:hnonder}
Let $n\in\Zbb_{\geq2}$, $u_*\in\{0,-\frac{2}{\beta}\}$, and $\beta$ an indeterminate. Then,
$h_{n,u_*}$ is non-derogatory if and only if $l_{u_*}^{(n)}=c_n$.
\end{prop}
\begin{proof}
For $\rho_n$ faithful, the minimal polynomial of $h_{u_*}$ is the same as that of $\rho_n(h_{u_*})$, 
and $h_{u_*}$ is non-derogatory if and only if $\rho_n(h_{u_*})$ is 
non-derogatory. The latter is also equivalent to $m_{\rho_n(h_{u_*})}=c_{\rho_n(h_{u_*})}$,
hence to $\deg(m_{\rho_n(h_{u_*})})=c_n$. 
Since $\rho_n$ is faithful for $\beta$ an indeterminate, the result follows.
\end{proof}

Through direct computation, we have found that the spectrum of $\rho_n(h_0)$ for $\beta=-2$ is non-degenerate
for $n=2,\ldots,17$. It follows that
\be
 l_{0,-2}^{(n)}=c_n,\qquad n=2,\ldots,17.
\label{l0}
\ee
We likewise find that the spectrum of $\rho_n(h_{n,\text{-}\frac{2}{\beta}})$ for $\beta=\pi+\pi^{-1}$ is 
non-degenerate for $n=2,\ldots,6$, hence
\be
 l_{\text{-}\frac{2}{\beta},\pi+\pi^{-1}}^{(n)}=c_n,\qquad n=2,\ldots,6.
\label{lb}
\ee
The specific $\beta$-values in these computations are immaterial, as long as they are `sufficiently generic'.
\begin{conj}\label{conj:DegenSpecHams}
For every $n\in\Zbb_{\geq2}$, each $u_*\in\{0,-\frac{2}{\beta}\}$, and $\beta$ an indeterminate, 
the spectrum of $\rho_n(h_{u_*})$ is non-degenerate.
\end{conj}
\noindent
This conjecture implies that, for every $n\in\Zbb_{\geq2}$ and $u_*\in\{0,-\frac{2}{\beta}\}$, we have
\be
 l_{u_*}^{(n)}=c_n.
\label{lbin}
\ee

\noindent
\textbf{Remark.}
Necessary conditions for {\em strict} inequalities in (\ref{ineq}) are the existence of spurious respectively 
permanent degeneracies in the spectrum of $\rho_n(h_{u_*})$. However, these are not {\em sufficient}
conditions as the Jordan-block structure may `prevent' a corresponding reduction in the degree of the minimal 
polynomials.
\begin{prop}\label{prop:SpurDegenFiniteh1}
Let $n\in\Zbb_{\geq2}$ and $u_*\in\{0,-\frac{2}{\beta}\}$. Then, there exist at most finitely many values 
$\beta\in\Cbb$ for which the spectrum of $\rho_n(h_{u_*})$ possesses spurious degeneracies. 
\end{prop}
\begin{proof}
Since the matrix elements of $\rho_n(h_{u_*})$ are polynomial in $\beta$, the result follows from
Proposition \ref{prop:SpDegenAx}.
\end{proof}
\begin{cor}
Let $n\in\Zbb_{\geq2}$ and $u_*\in\{0,-\frac{2}{\beta}\}$.
Then, there exist at most finitely many $\beta$-values for which
\be
 l_{u_*,\beta}^{(n)}<l_{u_*}^{(n)}.
\ee
\end{cor}
\begin{proof}
The result follows from Proposition \ref{prop:llbin} and Proposition \ref{prop:SpurDegenFiniteh1}.
\end{proof}
About the Jones--Wenzl idempotent, $\wj_n$, we note that (\ref{wj}) and (\ref{hB}) imply
$\wj_nh_{u_*}=h_{u_*}\wj_n=0$, hence $\wj_n\in\Cc_{\mathrm{TL}_{n}(\beta)}(h_{u_*})$.
Assuming (\ref{lbin}) holds, Proposition \ref{prop:hnonder} then implies that
$\wj_n\in\langle h_{u_*}\rangle_{\mathrm{TL}_{n}(\beta)}$ for $\beta$ an indeterminate.
It follows that, for every $n\in\Zbb_{\geq2}$ and each $u_*\in\{0,-\frac{2}{\beta}\}$, 
there exists polynomial $p_{u_*}^{(n)}$ such that $p_{u_*}^{(n)}(0)=1$ and
\be
 \wj_n=p_{u_*}^{(n)}(h_{u_*}),\qquad
 m_{u_*}^{(n)}(h_{u_*})=h_{u_*}\wj_n.
\label{wp}
\ee
Its degree is thus given by
\be
 \deg(p_{u_*}^{(n)})=l_{u_*}^{(n)}-1=c_n-1.
\label{wpd}
\ee
By (\ref{l0}) and (\ref{lb}), the relations (\ref{wp}) and (\ref{wpd}) do indeed hold for
$n=2,\ldots,17$ in the case $u_*=0$, and for $n=2,\ldots,6$ in the case $u_*=-\frac{2}{\beta}$.

\subsection{Transfer-operator hamiltonian polynomials}
\label{Sec:PolTL}

The above spectral analysis of the principal hamiltonians $h_{n,u_*}$ for small $n$, 
together with Proposition \ref{prop:spech} below, indicates that the hamiltonians are viable candidates for 
an integral of motion in terms of which the Temperley--Lieb transfer operator $T_n(u,\beta)$ is polynomial. 
We proceed by presenting explicit polynomial expressions of $T_n(u,\beta)$ in terms of the 
hamiltonians for small $n$, and by offering conjectures about the form of such polynomials for general $n$. 
\begin{prop}\label{prop:spech}
Let $n\in\Zbb_{\geq2}$, $u_*\in\{0,-\frac{2}{\beta}\}$, 
and $\psi$ a faithful representation of $\mathrm{TL}_{n}(\beta)$.
If the spectrum of $\psi(h_{u_*})$ is non-degenerate, then $T_n(u, \beta)$ is polynomial in $h_{u_*}$.
\end{prop}
\begin{proof}
Let the spectrum of $\psi(h_{u_*})$ be non-degenerate. Then, the characteristic and minimal polynomials of 
$\psi(h_{u_*})$ agree, so $\psi(h_{u_*})$ is non-derogatory. Since $\psi$ is faithful, it follows that $h_{u_*}$ is 
non-derogatory, and since $T_n(u,\beta)$ commutes with $h_{u_*}$, we have 
$T_n(u,\beta)\in \langle h_{u_*}\rangle_{\mathrm{TL}_n(\beta)}$.
\end{proof}
\noindent
\textbf{Remark.}
In the following, we will use that $\rho_n$ is faithful for $\beta$ an indeterminate and for all but finitely many 
$\beta\in\Cbb$.
\medskip

Because the matrix elements of $\rho_n(T_n(u, \beta))$ are polynomial in $\beta$, 
Proposition \ref{prop:SpDegenAx} implies that there are at most finitely many values
$\beta\in\Cbb$ for which the spectrum of $\rho_n(T_n(u, \beta))$ possesses spurious degeneracies. 
Combined with the non-degeneracy observations implying (\ref{l0}) and (\ref{lb}), it follows from 
Proposition \ref{prop:spech} that
for every $n=2,\ldots,17$, $T_n(u,\beta)$ is polynomial in $h_0$ for all but finitely many $\beta$-values,
and that for every $n=2,\ldots,6$, $T_n(u,\beta)$ is polynomial in $h_{\text{-}\frac{2}{\beta}}$ for all but finitely 
many $\beta$-values. For $n=2$, for example, we have $h_{\text{-}\frac{2}{\beta}}=-h_{0}$ for $\beta\neq0$, and
\be
  T_2(u,\beta)= \big(\beta[3]_{u^2} + 2u[2]_{u^2}\big)\mathds{1}_2-u(\beta + 2u)(2+\beta u)h_0,
\ee
valid for all $\beta\in\Cbb$. In the following, we will argue that $T_n(u, \beta)\in\Cbb[u][h_{u_*}]$, 
$u_*\in\{0,-\frac{2}{\beta}\}$, for every $n\in\Zbb_{\geq2}$ and all but finitely many $\beta$-values. We refer 
to these values as {\em $h_{u_*}$-exceptional} and note that the number and values of them will depend on 
$n$ and $u_*$.
\begin{conj}\label{conj:Ta}
Let $n\in\Zbb_{\geq3}$ and $\beta$ an indeterminate.
For each $u_*\in\{0,-\frac{2}{\beta}\}$, $T_n(u,\beta)$ admits a unique decomposition of the form
\be
  T_n(u,\beta) = \big(\beta[n+1]_{u^2} + 2u[n]_{u^2}\big)\idn
    +\frac{u(\beta + 2u)(2+\beta u)}{f_{n,u_*}(\beta)}
    \sum_{i=1}^{l_{u_*}^{(n)}-1}a^{n,u_*}_i(u,\beta) h_{u_*}^{i},
\label{equ:DTMHpolyRefined}
\ee
where $f_{n,u_*}(\beta)$ is a monic polynomial and $a^{n,u_*}_i(u,\beta)$ are polynomials 
such that no root of $f_{n,u_*}(\beta)$ is a root of $a^{n,u_*}_i(u,\beta)$ for all $i=1,\ldots,l_{u_*}^{(n)}-1$.
\end{conj}
\noindent
For $u_*=0$, we have verified Conjecture \ref{conj:Ta} for $n=3,4,5,6$, finding
\be
 f_{3,0}(\beta)=f_{4,0}(\beta)=1,\qquad
 f_{5,0}(\beta) = (\beta^2+4) (\beta^2-\tfrac{1}{2}) (\beta^4-27\beta^2+121),
\ee
and
\begin{align}
    f_{6,0}(\beta) =\; &\beta^2(\beta ^2+2)\big(\beta^2-\tfrac{1}{6}\big)\big(\beta^2-\tfrac{1}{4}\big)
     \big(\beta^2-\tfrac{5}{4}\big)(\beta ^2-3)\big(\beta^2-\tfrac{16}{3}\big)\big(\beta^2-\tfrac{121}{12}\big)
\nonumber\\[.1cm]
    &\big(\beta^4-14\beta^2+\tfrac{121}{9}\big)\big(\beta^4-25\beta^2-\tfrac{121}{2}\big)
      \big(\beta^6-\tfrac{41}{4}\beta^4+\tfrac{53}{2}\beta^2-\tfrac{9}{4}\big)(\beta^6-18\beta^4+81\beta^2-16)
\nonumber\\[.1cm]
    &\big(\beta^6+11\beta^4+\tfrac{185}{4}\beta^2+\tfrac{121}{2}\big)
      \big(\beta^8+\tfrac{19}{9}\beta^6+\tfrac{467}{18}\beta^4+\tfrac{320}{9}\beta^2-\tfrac{512}{9}\big),
\label{f6}
\end{align}
while for $u_*=-\frac{2}{\beta}$, we have verified it for $n=3,4,5$, finding
\begin{align}
 f_{3,\text{-}\frac{2}{\beta}}(\beta)
  =\;&\beta^6-64,
 \\[.1cm]
 f_{4,\text{-}\frac{2}{\beta}}(\beta)
  =\;&(\beta^8-256)^2(\beta^8+8\beta^6+32\beta^4+256)
   \big(\beta^8+\tfrac{8}{3}\beta^6-32\beta^4+\tfrac{256}{3}\beta^2-\tfrac{256}{3}\big)
  \nonumber\\[.1cm]
  &(\beta^{16}-8\beta^{14}-64\beta^{12}+640\beta^{10}+512\beta^8-26624\beta^6+98304\beta^4
   -131072\beta^2+65536),
 \\[.1cm]
  f_{5,\text{-}\frac{2}{\beta}}(\beta)
  =\;&(\beta^{10}-1024)^3p_5(\beta),
\end{align}
where $p_5(\beta)$ is a non-degenerate even polynomial of degree $198$. 
Explicit expressions for the associated 
polynomials $a^{n,u_*}_i(u,\beta)$ are not presented here. Instead, explicit expressions for similar 
polynomials will be provided in a refined formulation, Conjecture \ref{conj:Ttx}, in Appendix \ref{Sec:TL5}.
\begin{conj}
Let $n\in\Zbb_{\geq3}$ and $\beta$ an indeterminate. Then,
\be
 f_{n,\text{-}\frac{2}{\beta}}(\beta)=(\beta^{2n}-4^n)^{n-2}p_n(\beta),
\ee
where $p_n(\beta)$ is a non-degenerate even polynomial.
\end{conj}
If Conjecture \ref{conj:Ta} holds, then every $h_{u_*}$-exceptional $\beta$-value will be a root of 
$f_{n,u_*}(\beta)$. The converse need not be true since $T_n(u,\beta)$ could be polynomial in $h_{u_*}$ even if 
$\beta$ is a root of $f_{n,u_*}(\beta)$, see below. Letting $E_{n,u_*}$ denote the set of $h_{u_*}$-exceptional 
$\beta$-values and $Z_{n,u_*}$ the set of roots (or zeros) of $f_{n,u_*}(\beta)$, we thus have
\be
 E_{n,u_*}\subseteq Z_{n,u_*}.
\label{ER}
\ee
To appreciate what happens if this is not an equality, 
let $u\notin\{0,-\frac{\beta}{2},-\frac{2}{\beta}\}$ and rewrite (\ref{equ:DTMHpolyRefined}) as
\be
 \sum_{i=1}^{l_{u_*}^{(n)}-1}a^{n,u_*}_i(u,\beta) h_{u_*}^{i}
  =\frac{f_{n,u_*}(\beta)}{u(\beta + 2u)(2+\beta u)}\Big(T_n(u,\beta)-\big(\beta[n+1]_{u^2}+2u[n]_{u^2}\big)\idn\Big).
\ee
Specialising $\beta$ to a root, $\beta_r$, of $f_{n,u_*}(\beta)$ then means that
\be
 \sum_{i=1}^{l_{u_*}^{(n)}-1}a^{n,u_*}_i(u,\beta) h_{u_*}^{i}\big|_{\beta=\beta_r}=0,
\ee
and since at least one of the polynomials $a^{n,u_*}_i(u,\beta)$ is nonzero when evaluated at $\beta=\beta_r$, 
it follows that
\be
 l_{u_*,\beta_r}^{(n)}<l_{u_*}^{(n)}.
\ee

For $\beta_r\in Z_{n,u_*}\setminus E_{n,u_*}$, the decomposition (\ref{equ:DTMHpolyRefined}) is replaced by
\be
  T_n(u,\beta_r) = \big(\beta_r[n+1]_{u^2} + 2u[n]_{u^2}\big)\idn
    +u(\beta_r+ 2u)(2+\beta_r u)
    \!\sum_{i=1}^{l_{u_*,\beta_r}^{(n)}-1}\!a^{n,u_*}_{i,\beta_r}(u) h_{u_*}^{i}\big|_{\beta=\beta_r},
\ee
where $a^{n,u_*}_{i,\beta_r}(u)$ is polynomial for all $i$.
Although $Z_{3,0}=Z_{4,0}=\emptyset$ and $Z_{5,0}=E_{5,0}$, 
we find that such a root $\beta_r$ does indeed exist for $n=6$, as
\be
 Z_{6,0}\setminus E_{6,0}=\{0\}.
\ee
Note that $\beta_r=0$ is the only degenerate root of $f_{6,0}(\beta)$.
Through direct computation, we likewise find that
\be
 Z_{n,\text{-}\frac{2}{\beta}}=E_{n,\text{-}\frac{2}{\beta}},
  \qquad n=3,4,5,
\ee
but have not managed to determine $Z_{n,\text{-}\frac{2}{\beta}}$ and $E_{n,\text{-}\frac{2}{\beta}}$ for $n\geq6$.

For $\beta_r\in E_{n,u_*}$, $T_n(u,\beta_r)$ is {\em not} expressible as a polynomial 
in $h_{u_*}$. However, observing that
\be
 Z_{n,0}\cap Z_{n,\text{-}\frac{2}{\beta}}=\emptyset,\qquad n=3,4,5,
\ee
we conclude that
\be
 E_{n,0}\cap E_{n,\text{-}\frac{2}{\beta}}=\emptyset,\qquad n=3,4,5.
\label{EnEnempty}
\ee
It follows that, for $n=3,4,5$ and every $\beta\in\Cbb$, $T_n(u,\beta)$ {\em is} polynomial 
in {\em at least one} of the two hamiltonians: $h_0$ and $h_{\text{-}\frac{2}{\beta}}$.

\section{Discussion} 
\label{Sec:Discussion}

We have developed a general framework for describing integrable models based on planar algebras, and we 
have revisited the notion of integrals of motion from an algebraic perspective, introducing polynomial integrability
as a fundamental characteristic. We applied the framework to the Temperley--Lieb loop model and to 
an eight-vertex model, and discussed their polynomial integrability.

It would be interesting to explore further what polynomial integrability can teach us about an integrable model
-- in general but also in specific models of particular physical relevance.
For example, we would like to understand how polynomial integrability might recast
$T$-systems \cite{PK91,KP92}, $Y$-systems \cite{Zamolodchikov91,Zamolodchikov91b}
and functional relations more generally \cite{KNS11}, including the inversion relations
for critical dense polymers described by $\mathrm{TL}_n(\beta=0)$ \cite{PR07,PRT14b},
and the similar higher-order relations in \cite{MDPR14}.

In lattice-model language, our transfer operators are constructed on the strip. When extending to the 
cylinder or annulus, transfer operators may be constructed from \textit{affine tangles}, in which case the 
operators are morphisms in the \textit{affine category} of a given planar algebra \cite{J01, Ghosh11}.
We expect that most of our polynomial integrability considerations carry over to the periodic case.

On both the strip and cylinder, the polynomial integrability of the Temperley--Lieb loop model seems to have 
intricate implications for our understanding and analysis of the continuum scaling limit. If indeed the hamiltonian 
generates its own centraliser for {\em all} finite $n$, it would require reconciliation with established
insight \cite{KS94} into how the conformal integrals of motion \cite{BLZ96} inherit behaviour displayed by their 
finite-size counterparts.

Preliminary results indicate that properties and results similar to the ones reported here for the Temperley--Lieb 
loop model, and for the special eight-vertex model in Appendix \ref{Sec:Tensor}, apply for the standard 
six- and eight-vertex models \cite{Lieb67,FW70,Sutherland70,Baxter71,Baxter07}, RSOS 
models \cite{ABF84,FB85}, the dilute loop models related to the $O(n)$ models \cite{Nienhuis90,Nienhuis90b},
fused \cite{BR89} Temperley--Lieb loop models \cite{FR02,ZJ07,PRT14,MDPR14}, 
and models built on Brauer \cite{Brauer37} and Birman--Wenzl--Murakami algebras \cite{BW89,Mur87}.
We hope to return elsewhere with more on this.

\subsection*{Acknowledgements}

XP is supported by an Australian Postgraduate Award from the Australian Government.
JR is supported by the Australian Research Council under the Discovery Project scheme, 
project number DP200102316.
The authors thank Jon Links, Alexi Morin-Duchesne, Paul Pearce, Eric Ragoucy  and Yvan Saint-Aubin 
for discussions and comments. 
We would also like to thank the reviewers for their thoughtful suggestions.

\appendix

\section{Planar algebras}
\label{Sec:PAs}

\subsection{Naturality}
\label{Sec:Naturality}

Planar tangles can be {\em glued} or composed as follows.
Let $T$ and $S$ be planar tangles such that $\eta(D_0^S)=\eta(D)$ for some $D\in\mathcal{D}_T$.
Then, it is possible to isotopically deform $S$ so that it fits in the space occupied by $D$, 
in such a way that the nodes match up and the marked intervals are aligned. 
The gluing is then performed by replacing
$D$ with $S$ and removing the boundary intervals of $S$, including the associated marking, 
thereby producing the planar tangle denoted by $T\circ_D S$. 
To illustrate, the following tangle $S$ can be glued inside the tangle $T$:
\be
 T = \raisebox{-1.25cm}{\CompTangT}, \qquad\quad 
 S = \raisebox{-1.25cm}{\CompTangS}, \qquad\quad 
 T\circ_D S = \raisebox{-1.5cm}{\CompTangTS}.
\label{equ:CompositionTS}
\ee
Consistency between the composition of planar tangles and the action of the tangles as multilinear 
maps \eqref{PT} is often referred to as {\em naturality} and corresponds to 
\begin{align}\label{equ:NatSimpExII}
    \Pb_{T\circ_D S}(v_t, v_s)=\Pb_T(v_t, \Pb_S(v_s) ), \qquad
    v_t\in\!\bigtimes_{D_t\in\Dc_T\setminus\{D\}}\!P_{\eta(D_t)}, \quad 
    v_s\in\!\bigtimes_{D_s\in\Dc_S}\!P_{\eta(D_s)},
\end{align}
see \cite{JonesNotes, Poncini23} for more details.

\subsection{Unitality}
\label{Sec:Unitality}

Let $(P_n)_{n\in\mathbb{N}_0}$ be a planar algebra.
\begin{lem}\label{lem:id_and_rot_invertible}
    If $P_n$ has no nonzero null vectors, then $\mathsf{P}_{r_{n,0}}$ is the identity operator.
\end{lem}
\begin{proof}
    Let $v\in P_n$ and $T$ be a planar tangle for which $\Dc_T=\{D\}$ with $\eta(D)=n$. 
    By naturality, we then have
    \be
     \mathsf{P}_{T\circ_D\,r_{n,0}}(v) = \mathsf{P}_{T}(\mathsf{P}_{r_{n,0}}(v)),
    \ee
    hence
    \be
       \mathsf{P}_{T}(v - \mathsf{P}_{r_{n,0}}(v)) = 0,
    \ee
    so $v - \mathsf{P}_{r_{n,0}}(v)\in \mathrm{ker}(\mathsf{P}_T)$.
    Since $P_n$ has no nonzero null vectors, it follows that $\mathsf{P}_{r_{n,0}}(v) = v$ for all $v\in P_n$.
\end{proof}
\noindent
Together with naturality, Lemma \ref{lem:id_and_rot_invertible} implies the following result.
\begin{cor}
 If $P_n$ has no nonzero null vectors, then $\mathsf{P}_{r_{n,k}}$ is invertible for every $k\in\Zbb$.
\end{cor}

Now, let $(A_n)_{n\in\mathbb{N}_0}$ be a planar algebra of the type used in the paragraph 
containing \eqref{equ:MultPlanar}, and recall $\idn:= \Pb_{\mathrm{Id}_n}\!()$ (\ref{Aunit}).
\begin{prop}\label{prop:AssocUnitAlg}
Let $A_n$ be endowed with the multiplication induced by $M_n$, and suppose
$A_n$ has no nonzero null vectors. Then, $A_n$ is unital, with unit $\mathds{1}_n$.
\end{prop}
\begin{proof}
    Let $D_l$ and $D_u$ denote the lower respectively upper disk in the planar tangle $M_n$, and let $v\in A_n$. 
    Then, naturality implies that
    \begin{align}
  v\mathds{1}_n
  =\mathsf{P}_{M_n}(v, \mathds{1}_n)
  =\mathsf{P}_{M_n\circ_{D_u} \mathrm{Id}_n}(v)
  =\mathsf{P}_{r_{2n,0}}(v)
  =\mathsf{P}_{M_n\circ_{D_l} \mathrm{Id}_n}(v)
  =\mathsf{P}_{M_n}(\mathds{1}_n,v)
  =\mathds{1}_nv.
    \end{align}
By Lemma \ref{lem:id_and_rot_invertible}, $\mathsf{P}_{r_{2n,0}}(v)=v$, 
hence $v\mathds{1}_n=v=\mathds{1}_nv$.
\end{proof}

\section{\texorpdfstring{$\mathrm{TL}_n(\beta)$}{TLn(beta)} polynomials}
\label{Sec:App}

\subsection{Principal hamiltonian \texorpdfstring{$h_0$}{h0}}
\label{Sec:Apph0}

For $n=2,\ldots,7$ and $\beta$ an indeterminate, the minimal polynomial for $h_{0}$ is given by
\begin{align}
 m_{0}^{(2)}(h)&=h^2+\beta h,
\label{m2h}
\\[.2cm]
 m_{0}^{(3)}(h)&=h^3+2\beta h^2+(\beta^2-1)h,
\\[.2cm]
 m_{0}^{(4)}(h)&=h^6+6\beta h^5
  +14\big(\beta^2\!-\!\tfrac{2}{7}\big)h^4
  +16\beta\big(\beta^2\!-\!\tfrac{7}{8}\big)h^3
  +9\big(\beta^4\!-\!\tfrac{16}{9}\beta^2\!+\!\tfrac{4}{9}\big)h^2
  +2\beta\big(\beta^4\!-\!3\beta^2\!+\!2\big)h,
\\[.2cm]
 m_{0}^{(5)}(h)
 &=h^{10}
  +12\beta h^9
  +63\big(\beta^2-\tfrac{1}{7}\big)h^8
  +190\beta\big(\beta^2-\tfrac{41}{95}\big)h^7
  +\dots,
\label{m5h}
\\[.2cm]
 m_{0}^{(6)}(h)
 &=h^{20}
  +30\beta h^{19}
  +423\big(\beta^2-\tfrac{8}{141}\big)h^{18}
  +3726\beta\big(\beta^2-\tfrac{106}{621}\big)h^{17}
  +\dots,
\\[.2cm]
 m_{0}^{(7)}(h)
 &=h^{35}
  +60\beta h^{34}
  +1740\big(\beta^2-\tfrac{5}{174}\big)h^{33}
  +32488\beta\big(\beta^2-\tfrac{701}{8122}\big)h^{32}
  +\dots.
\end{align}
In a matrix representation of $\rho_n(h_{0})$, the off-diagonal elements are independent of $\beta$, whereas
the diagonal elements are of the form $-i\beta$, $i\in\{0,\ldots,\lfloor\frac{n}{2}\rfloor\}$. Since the 
number of elements equal to $-i\beta$ is $\binom{\lfloor\frac{n}{2}\rfloor}{i}\binom{\lceil\frac{n}{2}\rceil}{i}$, and
\be
 \sum_{i=0}^{\lfloor\frac{n}{2}\rfloor}\binom{\lfloor\frac{n}{2}\rfloor}{i}\binom{\lceil\frac{n}{2}\rceil}{i}i
 =\big\lfloor\tfrac{n}{2}\big\rfloor c_{n-1},
\ee
it follows that
\be
 m_0^{(n)}(h)=h^{c_n}
  +\big\lfloor\tfrac{n}{2}\big\rfloor c_{n-1}\beta h^{c_n-1}+\dots.
\label{mnh0}
\ee
We also note that the degree of the monic $\beta$-polynomial multiplying $h^i$ in $m_{0}^{(n)}(h)$ is given by
$c_n-i$, and that this $\beta$-polynomial is even (respectively odd) if the degree is even (respectively odd).

For $\beta=0$ and $n\geq2$, there are spurious degeneracies in the spectrum of $\rho_n(h_0)$, so
$l_{0,0}^{(n)}$ could be smaller than $l_{0}^{(n)}$.
Through direct computation, we find
\begin{align}
 m_{0}^{(n)}(h)\big|_{\beta=0}&=m_{0,0}^{(n)}(h),\qquad n=2,3,4,
\\[.15cm]
 m_{0}^{(5)}(h)\big|_{\beta=0}&=h\,m_{0,0}^{(5)}(h),
\\[.15cm]
 m_{0}^{(6)}(h)\big|_{\beta=0}&=h^2m_{0,0}^{(6)}(h),
\\[.15cm]
 m_{0}^{(7)}(h)\big|_{\beta=0}&=h^2U_6\big(\tfrac{h}{2}\big)m_{0,0}^{(7)}(h).
\end{align}
We thus have
\begin{center}
{\renewcommand{\arraystretch}{1.5}
\begin{tabular}{|c||c|c|c|c|c|c|c|c|c|}
 \hline
 $n$ & $2$ & $3$ & $4$ & $5$ & $6$ & $7$
 \\
\hline
\hline
 $l_{0,0}^{(n)}$ & $2$ & $3$ & $6$ & $9$ & $18$ & $27$
 \\
\hline
 $l_{0}^{(n)}$ & $2$ & $3$ & $6$ & $10$ & $20$ & $35$
 \\
\hline
\end{tabular}
}
\end{center}
which confirms the following conjecture for $n=2,\ldots,7$.
\begin{conj}
For $n\in\Zbb_{\geq2}$, we have
\be
 l_{0,0}^{(n)}=\tfrac{1}{2}\big(3^{\lfloor\frac{n+1}{2}\rfloor}+(-1)^n3^{\lfloor\frac{n-1}{2}\rfloor}\big).
\ee
\end{conj}

\subsection{Principal hamiltonian \texorpdfstring{$h_{\text{-}\frac{2}{\beta}}$}{halt}}
\label{Sec:Apphbeta}

For $n=2,3,4,5$, the principal hamiltonian $h_{n,\text{-}\frac{2}{\beta}}$ is given by
\begin{align}
 h_{2,\text{-}\frac{2}{\beta}}&=e_1,
\\[.2cm]
 h_{3,\text{-}\frac{2}{\beta}}&= -\tfrac{1}{2\beta}(\beta^2 + 4)(e_1 + e_2) +e_1e_2+e_2e_1,
\\[.2cm]
 h_{4,\text{-}\frac{2}{\beta}}&= \tfrac{1}{4\beta^2}(\beta^4+4\beta^2+16)(e_1 + e_3) 
  +\tfrac{1}{4\beta^2}(\beta^2+4)^2e_2
  -\tfrac{1}{2\beta}(\beta^2+4)(e_1e_2+e_2e_1+e_2e_3+e_3e_2)
\nonumber\\[.2cm]
 &-\tfrac{4}{\beta} e_1e_3 +
  e_1e_2e_3+e_3e_2e_1+e_1e_3e_2+e_2e_1e_3,
\\[.2cm]
 h_{5,\text{-}\frac{2}{\beta}}
 &= -\tfrac{1}{8\beta^3} (\beta ^2+4)(\beta ^4+16)(e_1+e_4) 
  -\tfrac{1}{8\beta^3} (\beta ^2+4)(\beta^4+4\beta^2+16)(e_2+e_3) 
\nonumber\\[.2cm]
 &+\tfrac{1}{4\beta^2}(\beta^4+4\beta^2+16)(e_1e_2 + e_3e_4 + e_4e_3 + e_2e_1)
  +\tfrac{1}{4\beta^2}(\beta ^2+4)^2 (e_2e_3 + e_3e_2)
\nonumber\\[.2cm]
 & +\tfrac{2}{\beta^2}(\beta ^2+4) (e_1e_3 + e_2e_4)
  +\tfrac{8}{\beta^2} e_1e_4
  -\tfrac{4}{\beta}(e_1e_2e_4 + e_1e_3e_4 + e_1e_4e_3 + e_2e_1e_4)
 \\[.2cm]
 &-\tfrac{1}{2\beta} (\beta ^2+4) (e_1e_2e_3 + e_1e_3e_2 + e_2e_1e_3 + e_2e_3e_4 + e_2e_4e_3 + e_3e_2e_1 + e_3e_2e_4 + e_4e_3e_2)
\nonumber\\[.2cm]
 &+e_1e_2e_3e_4 + e_1e_2e_4e_3 + e_1e_3e_2e_4 + e_1e_4e_3e_2 + e_2e_1e_3e_4 + e_2e_1e_4e_3 + e_3e_2e_1e_4 + e_4e_3e_2e_1,
\nonumber
\end{align}
and for $\beta$ an indeterminate, its minimal polynomial is given by
\begin{align}
 m_{\text{-}\frac{2}{\beta}}^{(2)}(h)&=h^2-\beta h,
\\[.2cm]
 m_{\text{-}\frac{2}{\beta}}^{(3)}(h)
 &=h^3+\tfrac{2\beta}{2\beta}(\beta^2+2)h^2+\tfrac{1}{(2\beta)^2}(\beta^6+3\beta^4+12\beta^2-16)h,
\\[.2cm]
 m_{\text{-}\frac{2}{\beta}}^{(4)}(h)
  &=h^6
  -\tfrac{6\beta}{(2\beta)^2}\big(\beta^4+\tfrac{8}{3}\beta^2+\tfrac{16}{3}\big)h^5
  +\tfrac{14}{(2\beta)^4}\big(\beta^{10}+\tfrac{38}{7}\beta^8 
  +\dots\big)h^4
  -\tfrac{16\beta}{(2\beta)^6}\big(\beta^{14}+\ldots\big)h^3+\dots,
\\[.2cm]
 m_{\text{-}\frac{2}{\beta}}^{(5)}(h)
  &=h^{10}
  +\tfrac{12\beta}{(2\beta)^3}(\beta^6+3\beta^4+8\beta^2+16)h^9
  +\tfrac{63}{(2\beta)^6}(\beta^{14}+\dots)h^8
  +\tfrac{190\beta}{(2\beta)^9}(\beta^{20}+\dots)h^7+\dots.
\end{align}
We note that the numerators of the fractions multiplying the even monic $\beta$-polynomials in these minimal 
polynomials are the same as the coefficients to the similar terms in (\ref{m2h})--(\ref{m5h}).
We also note that the degree of the monic $\beta$-polynomial multiplying $h^i$ in $m_{0}^{(n)}(h)$ is given by
$(n-1)(c_n-i)$, and that this $\beta$-polynomial is even (respectively odd) if the degree is even (respectively odd).
For the $h_{\text{-}\frac{2}{\beta}}$-pendant to (\ref{mnh0}), we conjecture the following expression.
\begin{conj}
For $n\in\Zbb_{\geq2}$, we have
\be
 m_{\text{-}\frac{2}{\beta}}^{(n)}(h)=h^{c_n}
  -(-1)^n\Big\lfloor\frac{n}{2}\Big\rfloor \frac{c_{n-1}}{2^{n-2}}
  \Big(\beta^{n-1}+\frac{4(n-2)}{n-1}\beta^{n-3}+\dots\Big)
  h^{c_n-1}+\dots.
\ee
\end{conj}

\subsection{Decomposition conjectures}
\label{Sec:TL5}

\begin{conj}\label{conj:Ttx}
Let $n\in\Zbb_{\geq3}$ and $\beta$ an indeterminate.
Then, $\Tt_n(x)$ admits a unique decomposition of the form
\be
 \Tt_n(x)=\big[\beta U_n(\tfrac{x}{2})+2U_{n-1}(\tfrac{x}{2})\big]\idn
 +(\beta^2+2\beta x+4)\Big(x^{n-3}\big[(\beta-x)h_0+h_0^2\big]
 +\frac{1}{f_{n,0}(\beta)}\sum_{i=1}^{c_n-1}
  \sum_{k=0}^{n-4}\at^{n,0}_{i,k}(\beta)x^kh_0^i\Big),
\label{Ttx}
\ee
where $f_{n,0}(\beta)$ is as in Conjecture \ref{conj:Ta} and $\at^{n,0}_{i,k}(\beta)$ are polynomials 
such that no root of $f_{n,0}(\beta)$ is a root of $\at^{n,0}_{i,k}(\beta)$ for all $i,k$.
\end{conj}
\noindent
The form of the contribution $x^{n-3}\big[(\beta-x)h_0+h_0^2\big]$
follows from continuing the expansion (\ref{Tneps}) to third order in $\epsilon$:
\begin{align}
 T_n(\epsilon,\beta)
 &=\big[\beta+2\epsilon-(n-1)\epsilon^2\beta-2(n-2)\epsilon^3\big]\idn
\nonumber\\[.1cm]
 &\quad-2\epsilon\beta h_0
 +\epsilon^2\big[2\beta h_0^2+(\beta^2-4)h_0\big]
 +\epsilon^3(4+\beta^2)(h_0^2+\beta h_0)
 +\Oc(\epsilon^4).
\end{align}
We have verified Conjecture \ref{conj:Ttx} for $n=3,4,5,6$, finding
\be
    \at_{1,0}^{4,0}(\beta) =\tfrac{1}{2}\beta^4-2\beta^2+2,\quad
    \at_{2,0}^{4,0}(\beta) = \tfrac{7}{4}\beta^3-\tfrac{7}{2}\beta,\quad
    \at_{3,0}^{4,0}(\beta)= \tfrac{9}{4}\beta^2-\tfrac{3}{2},\quad
    \at_{4,0}^{4,0}(\beta) = \tfrac{5}{4}\beta,\quad
    \at_{5,0}^{4,0}(\beta) = \tfrac{1}{4},
\ee
and
\begin{align}
  \at_{1,0}^{5,0}(\beta)&=16\beta^{13}-20\beta^{11}-1266\beta^9+\tfrac{20349}{4}\beta^7
    -\tfrac{16291}{4}\beta^5-\tfrac{13207}{4}\beta^3+1749\beta,
\\[.2cm]
  \at_{1,1}^{5,0}(\beta)&=-72\beta^{12}+\tfrac{1085}{2}\beta^{10}-\tfrac{1541}{4}\beta^8-\tfrac{29325}{8}\beta^6
     +\tfrac{49671}{8}\beta^4-\tfrac{3815}{8}\beta^2-\tfrac{1331}{2},
\\[.2cm]
  \at_{2,0}^{5,0}(\beta)&=104\beta^{12}-\tfrac{57}{2}\beta^{10}-7996\beta^8+\tfrac{46665}{2}\beta^6
   -\tfrac{21031}{2}\beta^4-\tfrac{16799}{2}\beta^2+1320,
\\[.2cm]
  \at_{2,1}^{5,0}(\beta)&=-460\beta^{11}+\tfrac{11439}{4}\beta^9+\tfrac{1901}{4}\beta^7-\tfrac{73429}{4}\beta^5
    +\tfrac{72761}{4}\beta^3-\tfrac{3277}{2}\beta,
\\[.2cm]
  \at_{3,0}^{5,0}(\beta)&=292\beta^{11}+\tfrac{599}{4}\beta^9-\tfrac{88045}{4}\beta^7+\tfrac{178239}{4}\beta^5
     -\tfrac{31661}{4}\beta^3-7096\beta,
\\[.2cm]
  \at_{3,1}^{5,0}(\beta)&=-1270\beta^{10}+\tfrac{52157}{8}\beta^8+\tfrac{53071}{8}\beta^6-\tfrac{291349}{8}\beta^4
    +\tfrac{152711}{8}\beta^2-\tfrac{627}{2},
\\[.2cm]
  \at_{4,0}^{5,0}(\beta)&=462\beta^{10}+518\beta^8-\tfrac{68991}{2}\beta^6+\tfrac{89539}{2}\beta^4
   +\tfrac{1515}{2}\beta^2-2134,
\\[.2cm]
  \at_{4,1}^{5,0}(\beta)&=-1977\beta^9+\tfrac{16893}{2}\beta^7+15452\beta^5-\tfrac{73257}{2}\beta^3+8347\beta,
\\[.2cm]
  \at_{5,0}^{5,0}(\beta)&=450\beta^9+\tfrac{2813}{4}\beta^7-33643\beta^5+24358\beta^3+3106\beta,
\\[.2cm]
  \at_{5,1}^{5,0}(\beta)&=-1896\beta^8+\tfrac{54903}{8}\beta^6+17612\beta^4-19783\beta^2+1254,
\\[.2cm]
  \at_{6,0}^{5,0}(\beta)&=276\beta^8+\tfrac{1021}{2}\beta^6-\tfrac{41757}{2}\beta^4+6328\beta^2+924,
\\[.2cm]
  \at_{6,1}^{5,0}(\beta)&=-1146\beta^7+\tfrac{14491}{4}\beta^5+\tfrac{45761}{4}\beta^3-5432\beta,
\\[.2cm]
  \at_{7,0}^{5,0}(\beta)&=104\beta^7+\tfrac{829}{4}\beta^5-\tfrac{32131}{4}\beta^3+315\beta,
\\[.2cm]
  \at_{7,1}^{5,0}(\beta)&=-426\beta^6+\tfrac{9827}{8}\beta^4+\tfrac{34447}{8}\beta^2-\tfrac{1177}{2},
\\[.2cm]
  \at_{8,0}^{5,0}(\beta)&=22\beta^6+44\beta^4-1747\beta^2-110,
\\[.2cm]
  \at_{8,1}^{5,0}(\beta)&=-89\beta^5+247\beta^3+\tfrac{1737}{2}\beta,
\\[.2cm]
  \at_{9,0}^{5,0}(\beta)&=2\beta^5+\tfrac{15}{4}\beta^3-164\beta,
\\[.2cm]
  \at_{9,1}^{5,0}(\beta)&=-8\beta^4+\tfrac{181}{8}\beta^2+\tfrac{143}{2}.
\end{align}
Although the polynomials $\at_{i,k}^{6,0}(\beta)$ are not provided here, we note that,
for $n=4,5,6$ and all $i,k$,
\be
  \deg(\at^{n,0}_{i,k})=d_{n,0}-i-\tfrac{1}{2}\big(1-(-1)^k\big),\qquad d_{4,0}=5,\ d_{5,0}=14,\ d_{6,0}=63,
\ee
and that $\at_{i,k}^{n,0}(\beta)$ is even (respectively odd) if its degree is even (respectively odd).
This is seen to correspond to the parity of $n+i+k+1$.
\begin{conj}\label{conj:Ttxb}
Let $n\in\Zbb_{\geq3}$ and $\beta$ an indeterminate.
Then, $\Tt_n(x)$ admits a unique decomposition of the form
\be
 \Tt_n(x)=\big[\beta U_n(\tfrac{x}{2})+2U_{n-1}(\tfrac{x}{2})\big]\idn
 +\frac{\beta^2+2\beta x+4}{f_{n,\text{-}\frac{2}{\beta}}(\beta)}
  \sum_{i=1}^{c_n-1}
  \sum_{k=0}^{n-2}\at^{n,\text{-}\frac{2}{\beta}}_{i,k}(\beta)x^kh_{\text{-}\frac{2}{\beta}}^i,
\label{Ttxb}
\ee
where $f_{n,\text{-}\frac{2}{\beta}}(\beta)$ is as in Conjecture \ref{conj:Ta} 
and $\at^{n,\text{-}\frac{2}{\beta}}_{i,k}(\beta)$ are polynomials such that no root of 
$f_{n,\text{-}\frac{2}{\beta}}(\beta)$ is a root of $\at^{n,\text{-}\frac{2}{\beta}}_{i,k}(\beta)$ for all $i,k$.
\end{conj}
\noindent
We have verified Conjecture \ref{conj:Ttxb} for $n=3,4,5$, where we note that, for all $i,k$,
\be
  \deg(\at^{n,\text{-}\frac{2}{\beta}}_{i,k})=d_{n,\text{-}\frac{2}{\beta}}-(n-1)i-k,\qquad 
  d_{3,\text{-}\frac{2}{\beta}}=9,\ d_{4,\text{-}\frac{2}{\beta}}=54,\ d_{5,\text{-}\frac{2}{\beta}}=235,
\ee
and that $\at_{i,k}^{n,\text{-}\frac{2}{\beta}}(\beta)$ is even (respectively odd) if its degree is even 
(respectively odd). This is seen to correspond to the parity of $n+i+k$.

\section{Tensor planar algebras}
\label{Sec:Tensor}

Here, we specialise to {\em tensor planar algebras} and thereby recover the familiar quantum inverse scattering 
method framework, in which case the $R$-operators are tensorially separable,
and outline how the planar-algebraic framework simplifies.
To illustrate, we consider a specialisation of the zero-field eight-vertex model \cite{FW70, Sutherland70}
that satisfies the free-fermion condition, and whose principle hamiltonian corresponds to the Ising model 
hamiltonian. The general model was solved by Baxter in \cite{Baxter71}, see also \cite{Baxter07}, 
while our presentation highlights the underlying polynomial integrable structure of a particular specialisation, 
by analysing the spectral properties of the transfer operator and the associated hamiltonian. 

We thus show that the transfer matrix of this specialised eight-vertex model is diagonablisable and present an 
exact solution. Although the model is Yang--Baxter integrable, its simplicity allows us to use standard techniques
to obtain explicit expressions for all eigenvalues and corresponding eigenvectors of the transfer matrix.
We then exploit similarities in the spectral properties of the transfer matrix and the canonical hamiltonian to 
establish that the transfer matrix is polynomial in the hamiltonian.
Moreover, we decompose the transfer operator into an explicit linear combination
of a complete set of orthogonal idempotents expressed in terms of the minimal polynomial of the hamiltonian.

\subsection{Definition and cellularity}
\label{app:TensorPADef}

For each $n\in\Nbb_0$ and $\ell\in\Nbb$, we let $V_n$ denote the $\ell^{2n}$-dimensional vector space 
spanned by disks with $2n$ labelled boundary points where each label is taken from the set $\{1,\ldots,\ell\}$.
As the disks do not come equipped with any further (interior or otherwise) structure, 
we have $V_n\cong V^{\otimes 2n}$, where $V$ is an $\ell$-dimensional vector space. 

The \textit{tensor planar algebra} is the graded vector space $(V_n)_{n\in\Nbb_0}$, 
together with the following action of planar tangles:
If a string in the planar tangle connects a pair of boundary points with different labels, 
then the output is the zero vector, and if not, then the labels of the output vector
are given by the labels at the opposite string endpoints. 
If both endpoints of a string are on the exterior boundary of the planar tangle, then the output vector is a sum
obtained by varying the common label of the two endpoints.
Following from compatibility with the gluing of planar tangles, and using the evaluation map $\mathrm{e}$, 
a loop is accordingly replaced by the scalar $\ell$. 
To illustrate, with $T$ as in \eqref{equ:GenPlanarActionExample} and $a,b\in\Fbb$, we have
\begin{align}
    \Pb_T\Big(\!\!\!\raisebox{-0.4cm}{\RotPlanarExampleNodeIIIAltCTensor}\!,
    a\!\!\raisebox{-0.4cm}{\RotPlanarExampleNodeIAltCTensorAlt}\!\!\!
    +b\!\!\raisebox{-0.4cm}{\RotPlanarExampleNodeIAltCTensor}\!,\!\!\!
    \raisebox{-0.425cm}{\RotPlanarExampleNodeIIAltCTensor}\!\!\!\Big)
    =  a\raisebox{-1.45cm}{\RotPlanarNKOperandExampleAltCTensorAltNew}\!\! 
    + b\raisebox{-1.45cm}{\RotPlanarNKOperandExampleAltCTensorNew}\!\! 
    = b\,\ell\sum_{k=1}^{\ell}\!\!\raisebox{-0.485cm}{\RotPlanarExampleNodeIVAltCTensorTwo}\!.
\end{align}

We now equip each $V_n$ with the multiplication induced by the corresponding multiplication tangle, $M_n$
in \eqref{equ:MultPlanar}, and identify the first $n$ labels clockwise from the marked boundary interval 
as characterising an incoming vector, with the remaining $n$ labels characterising an outgoing vector.
The vector space $V_n$ thus has the structure of an endomorphism algebra,
\be
 V_n = \mathrm{End}_{\Fbb}(V^{\otimes n}),
\ee
and is consequently cellular.

\subsection{Transfer operator}
\label{subsec:TPAtransferop}

Let bases for $V_1$ and $V_2$ be given by
\be
 B_1=\{e_j^k\,|\,j,k\in\{1,\ldots,\ell\}\},\qquad
 B_2=\{e_{j\,l}^{k\,m}\,|\,j,k,l,m\in\{1,\ldots,\ell\}\},
\ee
respectively, such that,
viewed as matrices relative to the natural basis orderings, $e_j^k$ are $\ell\times\ell$ matrix units with $1$ 
in position $(j,k)$ and zeros elsewhere. By construction,
\be
 V_n\cong[\mathrm{End}_{\Fbb}(V)]^{\otimes n},
\ee
so every element of $V_2$ is separable. In particular, the $V_2$ basis vectors decompose as
\be
 e_{j\,l}^{k\,m}=e_j^k\otimes e_l^m,\qquad j,k,l,m\in\{1,\ldots,\ell\}.
\ee

As parameterised elements of $V_1$ and $V_2$, respectively, the $K$- and $R$-operators (\ref{RKK}) 
are here written as
\begin{align}
 K(u)=\sum_{j,k=1}^{\ell}K_{j}^k(u)\raisebox{-0.45cm}{\VecMatAlKop},\qquad
 R(u)= \!\!\!\sum_{j,k,l,m=1}^\ell\!\!\! R_{j\,l}^{k\,m}(u) \!\!\raisebox{-0.45cm}{\VecMatAlRop},\qquad
 \Ko(u)=\sum_{j,k=1}^{\ell}\Ko_{j}^k(u)\raisebox{-0.45cm}{\VecMatAlKop},
\end{align}
where $K_{j}^k,R_{j\,l}^{k\,m},\overline{K}_{j}^k$ are scalar functions, and where we have attached short strings 
to the labelled boundary points, for illustrative purposes. The separability of $V_2$ allows us to write
\begin{align}
    R(u) 
    = \!\!\!\sum_{j,k,l,m=1}^\ell\!\!\! R_{j\,l}^{k\,m}(u) \!\!\raisebox{-0.5cm}{\VecMatAlRopDecomp} 
    = \!\!\!\sum_{j,k,l,m=1}^\ell\!\!\! R_{j\,l}^{k\,m}(u)\,e_j^k \otimes e_l^m.
\end{align}
Using the {\em same} scalar functions but with the upper indices interchanged, we get
\begin{align}
    \check{R}(u) 
    = \!\!\!\sum_{j,k,l,m=1}^\ell\!\!\! R_{j\,l}^{k\,m}(u) \!\!\raisebox{-0.5cm}{\Rtwist} 
    = \!\!\!\sum_{j,k,l,m=1}^\ell\!\!\! R_{j\,l}^{k\,m}(u) \!\!\raisebox{-0.5cm}{\VecMatAlRopDecompTwo} 
    = \!\!\!\sum_{j,k,l,m=1}^\ell\!\!\! R_{j\,l}^{k\,m}(u)\,e_j^m \otimes e_l^k,
\end{align}
in terms of which we construct the Sklyanin transfer operator $\check{T}(u)\in V_n$, as in (\ref{TS}).

Accordingly, $\check{T}(u)$ can be expressed familiarly as a vector-space trace over an auxiliary copy of 
$\mathrm{End}_{\Fbb}(V)$ in $V_{n+1}$. In the following, the auxiliary space is the $(n+1)^{\mathrm{th}}$ copy,
and the corresponding trace is denoted by $\mathrm{Tr}_{n+1}$.
For each $i=1,\ldots,n$, we first introduce
\begin{align}
    R_{i,n+1}(u) & : =\!\!\!\sum_{j,k,l,m=1}^\ell\!\!\! R_{j\,l}^{k\,m}(u)\,\mathds{1}_{i-1}\otimes e_j^k
  \otimes\mathds{1}_{n-i}\otimes e_l^m,\\
    R_{n+1, i}(u) &: =\!\!\!\sum_{j,k,l,m=1}^\ell\!\!\! R_{j\,l}^{k\,m}(u)\,\mathds{1}_{i-1}\otimes 
  e_l^m\otimes\mathds{1}_{n-i}\otimes e_j^k,
\end{align}
where the indices $i$ and $n+1$ denote the copies of $V$\! on which the operators act nontrivially,
as well as
\be
 K_1(u): =\sum_{j,k=1}^{\ell}K_{j}^k(u)\,e_j^k\otimes\idn,\qquad
 \Ko_{n+1}(u): =\sum_{j,k=1}^{\ell}\Ko_{j}^k(u)\,\idn\otimes e_j^k.
\ee
The transfer operator can then be expressed in terms of the `pre-trace' transfer operator
\begin{align}
    L_{n+1}(u) : =R_{n,n+1}(u)
        \cdots R_{1,n+1}(u)K_1(u)R_{n+1,1}(u)\cdots R_{n+1,n}(u)\Ko_{n+1}(u)
\end{align}
as
\be
 \check{T}_n(u) = \mathrm{Tr}_{n+1}(L_{n+1}(u)).
\label{Tcheck}
\ee

\subsection{Eight-vertex model}
\label{subsec:8}

In the remainder of this appendix, we let
\be
 \mathrm{dim}(V) = 2
\ee
and fix the parameterisation to an eight-vertex model, characterised by
\begin{align}\label{equ:ParameterisationMatEx}
    R_{j\,l}^{k\,m}(u) = 
    \begin{cases}
        1, &\ j=m, l=k,\\
        u, &\ j=l, k=m, j\neq k,\\
        u, &\ j=k, l=m, j\neq l,\\
        0, &\ \mathrm{otherwise},
    \end{cases} \qquad\quad
     K_{j}^{k}(u) = \Ko_{j}^{k}(u)  = 
    \begin{cases}
        1, &\ j=k,\\
        0, &\ j\neq k.
    \end{cases}
\end{align}
Working in the natural matrix representation where the $e_j^k$ are matrix units, and where $\mathds{1}$ denotes 
the $2\times2$ identity matrix, the $R$- and $K$-operators can be expressed in terms of Pauli matrices as
\begin{align}
    R(u) = \frac{1}{2}\big[
    (1+u)(\mathds{1}\otimes \mathds{1}+\sigma^x\otimes\sigma^x) 
    + (1-u)(\sigma^y\otimes\sigma^y+\sigma^z\otimes\sigma^z)\big], \qquad K(u) = \Ko(u) = \mathds{1}.
\end{align}
It follows that
\begin{align}
    R^2(u) &= (1+u^2)\,\mathds{1}\otimes \mathds{1}+2u\,\sigma^x\otimes\sigma^x,
\label{equ:RSq8v}
\\[.2cm]
    R(u)(\mathds{1}\otimes \sigma^x)R(u) &= (1+u^2)\,\sigma^x\otimes \mathds{1}+2u\,\mathds{1}\otimes\sigma^x,
\label{equ:RXSq8v}
\end{align}
and using the standard notation
\be
 \sigma_{m,i}^{\alpha}: = \mathds{1}_{i-1}\otimes \sigma^\alpha \otimes \mathds{1}_{m-i},\qquad 
  \alpha\in\{x,y,z\},\ i\in\{1,\ldots,m\},\ m\in\Nbb,
\ee
we then have the following result.
\begin{lem}
$L_{n+1}(u)$ is polynomial in $\sigma_{n+1,1}^x,\ldots,\sigma_{n+1,n+1}^x$.
\end{lem}
\noindent
It follows that
\begin{align}\label{equ:SigmaXLn}
    L_{n+1}(u) = L_n^{(0)}(u)\otimes \mathds{1} + L_n^{(1)}(u)\otimes \sigma^x
\end{align}
for some $L_n^{(0)}(u),L_n^{(1)}(u)\in V_n$, and consequently that 
\begin{align}\label{equ:TransLop8v}
    \check{T}_{n}(u) = 2 L_{n}^{(0)}(u)
\end{align}
and
\be
 [\check{T}_{n}(u) ,\check{T}_{n}(v) ]=0,\qquad\forall\, u,v\in\Fbb.
\label{TcTc}
\ee

The next result allows us to determine the polynomial structures of $L_{n+1}(u)$
and $\check{T}_{n}(u)$.
\begin{prop}
With $L_1(u)\equiv\mathds{1}$ and $\check{T}_1(u)=2(1+u^2)\mathds{1}$, 
the matrices $L_{n+1}(u)$ and  $\check{T}_{n}(u)$ are determined recursively by
\begin{align}
 L_{n+1}(u)&=\big((1+u^2)\,\mathds{1}_{n+1}+2u\,\sigma^x\otimes\sigma^x\otimes\mathds{1}_{n-1}\big)
  \big(\mathds{1}\otimes L_n(u)\big),
\label{Lrec}
\\[.2cm]
 \check{T}_{n}(u)&=\big((1+u^2)\,\mathds{1}_{n}+2u\,\sigma^x\otimes\sigma^x\otimes\mathds{1}_{n-2}\big)
  \big(\mathds{1}\otimes \check{T}_{n-1}(u)\big).
\label{Trec}
\end{align}
\end{prop}
\begin{proof}
Using $K(u) = \Ko(u) = \mathds{1}$ and that $R_{n+1,i}(u) =R_{i,n+1}(u)$ for all $i$, 
the relation (\ref{Lrec}) follows from
\begin{align}
 L_{n+1}(u)&=R_{n,n+1}(u)\cdots R_{1,n+1}(u)R_{1,n+1}(u)\cdots R_{n,n+1}(u)
\nonumber\\[.2cm]
 &=R_{n,n+1}(u)\cdots R_{2,n+1}(u)\big[(1+u^2)\mathds{1}_{n+1}+2u\,\sigma^x
   \otimes\mathds{1}_{n-1}\otimes\sigma^x\big]R_{2,n+1}(u)\cdots R_{n,n+1}(u)
\nonumber\\[.2cm]
 &=(1+u^2)\mathds{1}\otimes L_n(u)+2u\,\sigma^x\otimes\big(R_{n-1,n}(u)\cdots 
   R_{1,n}(u)[\mathds{1}_{n-1}\otimes\sigma^x]R_{1,n}(u)\cdots R_{n-1,n}(u)\big)
\nonumber\\[.2cm]
 &=\big((1+u^2)\,\mathds{1}_{n+1}\big)\big(\mathds{1}\otimes L_n(u)\big)
\nonumber\\[.2cm]
  &\quad +2u\,\sigma^x\otimes\big(R_{n-1,n}(u)\cdots R_{2,n}(u)[\sigma^x\otimes\mathds{1}_{n-1}]
   R_{1,n}(u)R_{1,n}(u)\cdots R_{n-1,n}(u)\big)
\nonumber\\[.2cm]
 &=\big((1+u^2)\,\mathds{1}_{n+1}\big)\big(\mathds{1}\otimes L_n(u)\big)
  +2u\big(\sigma^x\otimes\sigma^x\otimes\mathds{1}_{n-1}\big)\big(\mathds{1}\otimes L_n(u)\big),
\end{align}
where the fourth equality is a consequence of
\be
 R(u)(\mathds{1}\otimes \sigma^x)R(u)=(\sigma^x\otimes\mathds{1})R^2(u).
\ee
The relation (\ref{Trec}) is an immediate consequence of (\ref{Lrec}).
\end{proof}
In preparation for giving explicit expressions for $\check{T}_n(u)$, let
\be
 I_{n,k}: =\{(i_1,\ldots,i_{2k})\in\Nbb^{2k}\,|\,1\leq i_1<\cdots<i_{2k}\leq n\},\qquad
 k=1,\ldots,\lfloor\tfrac{n}{2}\rfloor,
\ee
and define
\be
 \kappa_{n,k}\colon I_{n,k}\to\Zbb,\qquad (i_1,\ldots,i_{2k})\mapsto\sum_{l=1}^{2k}(-1)^li_l.
\ee
Note that $\kappa_{n,k}(\iota)\in\{1,\ldots,n-1\}$ for all $\iota\in I_{n,k}$.
\begin{prop}\label{prop:Tc}
The transfer matrix admits the multiplicative expression
\be
 \check{T}_n(u)=2(1+u^2)\prod_{i=1}^{n-1}\big[(1+u^2)\,\idn+2u\,\sigma_{n,i}^x\sigma_{n,i+1}^x\big]
\label{Tmult}
\ee
and the additive expression
\be
 \check{T}_n(u)=2(1+u^2)^n\idn
 +2(1+u^2)\sum_{k=1}^{\lfloor\frac{n}{2}\rfloor}\,\sum_{\iota=(i_1,\ldots,i_{2k})\in I_{n,k}}
 \!\!(1+u^2)^{n-\kappa_{n,k}(\iota)}(2u)^{\kappa_{n,k}(\iota)}\sigma_{n,i_1}^x\!\cdots\sigma_{n,i_{2k}}^x.
\label{Tadd}
\ee
\end{prop}
\begin{proof}
The multiplicative expression is readily seen to satisfy the recursion relation (\ref{Trec}), including the initial 
condition for $n=1$. The additive expression follows by expanding the multiplicative expression.
\end{proof}

We let $|\pm\rangle$ denote eigenvectors of $\sigma^x$,
\be
    \sigma^x |\pm\rangle = \pm |\pm\rangle, \qquad 
    |\pm\rangle: =\begin{bmatrix}
        1\\
        \pm 1
    \end{bmatrix},
\ee
and let
\be
 K_n\colon  \{\pm\}^{n} \to \{0,1,\ldots,n-1\},\qquad
 (s_1,\ldots,s_n)\mapsto\sum_{i=1}^{n-1}|s_i-s_{i+1}|,
\ee
denote the function that counts the number of sign changes present in $(s_1,\ldots,s_n)\in\{\pm\}^{n}$. 
The pre-images
\be
 \Vc_{n,k}: =\mathrm{span}_{\Fbb}\{|K_n^{-1}(k)\rangle\},\qquad k=0,1,\ldots,n-1,
\label{Vcalk}
\ee
have dimension
\be
 \dim\Vc_{n,k}=2\,\binom{n-1}{k},
\ee
consistent with
\be
 \dim(V^{\otimes n})=\big|\{\pm\}^{n}\big|=2^n=\sum_{k=0}^{n-1}\dim\Vc_{n,k}.
\ee

\begin{prop}\label{prop:Tsl}
The transfer matrix $\check{T}_n(u)$ is diagonalisable, with eigenvectors
\be
 |\mathbf{s}\rangle=|s_1\rangle\otimes\cdots\otimes|s_n\rangle,\qquad 
   \mathbf{s}=(s_1,\ldots,s_n)\in\{\pm\}^{n},
\ee
and corresponding eigenvalues
\be
 \lambda_{\mathbf{s}}(u)=2(1+u^2)(1-u)^{2K_n(\mathbf{s})}(1+u)^{2(n-1-K_n(\mathbf{s}))}.
\ee
\end{prop}
\begin{proof}
The result follows from the multiplicative expression (\ref{Tmult}).
\end{proof}
\noindent
The following result readily follows from Proposition \ref{prop:Tsl}.
\begin{cor}\label{cor:TVk}
For each $k=0,1,\ldots,n-1$, $\Vc_{n,k}$ is the $\check{T}_n(u)$-eigenspace corresponding to the 
eigenvalue $2(1+u^2)(1-u)^{2k}(1+u)^{2(n-1-k)}$, 
and we have the eigenspace decomposition
\be
 V^{\otimes n}=\bigoplus_{k=0}^{n-1}\Vc_{n,k}.
\ee
\end{cor}

\subsection{Polynomial integrability}
\label{Sec:8pol}

Since $\check{T}_{n}(u)$ is diagonalisable and satisfies (\ref{TcTc}), 
the eight-vertex model is polynomially integrable. 
As we show in the following, the transfer matrix is polynomial in the principal hamiltonian given in (\ref{hn8}).

It follows from Proposition \ref{prop:Tc} that $u_*=0$ is the only identity point, with
\be
 \check{T}_n(\epsilon)
  =2\idn+4\epsilon\sum_{i=1}^{n-1}\sigma^x_{n,i}\sigma^x_{n,i+1}+\Oc(\epsilon^2).
\ee
The corresponding (renormalised) principal hamiltonian is given by
\be
 h_n=\sum_{i=1}^{n-1}\sigma^x_{n,i}\sigma^x_{n,i+1}.
\label{hn8}
\ee
\begin{prop}\label{prop:hnk}
The principal hamiltonian $h_n$ is diagonalisable, with eigenvectors
\be
 |\mathbf{s}\rangle=|s_1\rangle\otimes\cdots\otimes|s_n\rangle,\qquad 
   \mathbf{s}=(s_1,\ldots,s_n)\in\{\pm\}^{n},
\ee
and corresponding eigenvalues
\be
 \mu_{\mathbf{s}} = n-1 -2K_n(\mathbf{s}).
\ee
Moreover, for each $k=0,1,\ldots,n-1$, $\Vc_{n,k}$ is the $h_n$-eigenspace corresponding to the eigenvalue
$n-1-2k$.
\end{prop}
\begin{proof}
The result follows from (\ref{hn8}).
\end{proof}
\noindent
\textbf{Remark.}
The eigenvalues of $\check{T}_n(u)$ and $h_n$ are related as
\be
 \lambda_{\mathbf{s}}(\epsilon)=2\big[1+2\epsilon\mu_{\mathbf{s}}+\Oc(\epsilon^2)\big],\qquad
  \forall\,\mathbf{s}\in\{\pm\}^{n}.
\ee

The form of the minimal polynomial of $h_n$ follows from Proposition \ref{prop:hnk}. 
To fix our notation, we define the polynomials
\be
 m_j(x): =\prod_{k=0}^{j-1}[x-(j-1-2k)],\qquad j\in\Nbb.
\ee
\begin{cor}\label{cor:hnmn}
The minimal polynomial of $h_n$ is $m_n$.
\end{cor}

Although $h_n$ is not non-derogatory, it follows from Corollary \ref{cor:TVk} and 
Proposition \ref{prop:hnk} that
\be
 \check{T}_n(u)\in\Fbb(u)[h_n].
\label{TCh}
\ee
The next result provides details of this polynomial.
In preparation, let $\lambda_k(u)$ and $\mu_k$ denote the 
eigenvalues corresponding to the (joint) eigenspaces $\Vc_{n,k}$,
$k=0,1,\ldots,n-1$, of $\check{T}_n(u)$ and $h_n$, respectively.
For ease of reference, we recall their expressions,
\be
 \lambda_k(u)=2(1+u^2)(1-u)^{2k}(1+u)^{2(n-1-k)},\qquad
 \mu_k=n-1-2k,
\ee
and introduce the $n\times n$ matrix
\be
 \mathsf{V}_n: =\begin{bmatrix}
            1 & \mu_0 & \ldots & \mu_0^{n-1}\\[.15cm]
            1 & \mu_1 & \ldots & \mu_1^{n-1}\\
            \vdots & \vdots & \ddots 
            & \vdots\\
            1 & \mu_{n-1} & \ldots & \mu_{n-1}^{n-1}\\
        \end{bmatrix},\qquad n\in\Nbb.
\label{van}
\ee
Since this is a Vandermonde matrix with $\mu_i\neq\mu_j$ for all $i\neq j$, it is invertible,
and the inverse can be evaluated explicitly.
For every $n\in\Nbb$, we let $m_0(h_n)\equiv h_n^0\equiv\idn$.
\begin{prop}
For every $n\in \Nbb$ and all $u\in\Fbb$, we have 
\be
   \check{T}_n(u) = \sum_{i=0}^{n-1}\tau_i(u)h_n^i,
\label{equ:PolyTuDiag8v}
\ee
where
\be
 \tau_i(u)=\sum_{j=0}^{n-1}[\mathsf{V}_n^{-1}]_{i+1,j+1}\lambda_j(u),\qquad
  i=0,1,\ldots,n-1.
\label{equ:PolyTuDiagCoeffs8v}
\ee
\end{prop}
\begin{proof}
That an expression of the form (\ref{equ:PolyTuDiag8v}) exists follows from (\ref{TCh}) and 
Corollary \ref{cor:hnmn}. Using the common eigenbasis $\{|\mathbf{s}\rangle |\,\mathbf{s}\in\{\pm\}^{n}\}$ 
of $\check{T}_n(u)$ and $h_n$ to diagonalise the expression, yields
\be
    \mathrm{diag}\big(\lambda_0(u),\lambda_1(u),\ldots,\lambda_{n-1}(u)\big)
     =\sum_{i=0}^{n-1}
       \mathrm{diag}\big(\tau_i(u)\mu_0^i,\tau_i(u)\mu_1^i,\ldots,\tau_i(u)\mu_{n-1}^i\big),
\ee
where we have omitted repeated eigenvalues. 
Compressing the diagonal matrices into vectors, we obtain
    \begin{align}
         \begin{bmatrix}
            \lambda_0(u)\\[.15cm]
            \lambda_1(u)\\
            \vdots\\
            \lambda_{n-1}(u)
        \end{bmatrix}
     =
        \begin{bmatrix}
            1 & \mu_0 & \ldots & \mu_0^{n-1}\\[.15cm]
            1 & \mu_1 & \ldots & \mu_1^{n-1}\\
            \vdots & \vdots & \ddots 
            & \vdots\\
            1 & \mu_{n-1} & \ldots & \mu_{n-1}^{n-1}\\
        \end{bmatrix}
               \begin{bmatrix}
            \tau_0(u)\\[.15cm]
            \tau_1(u)\\
            \vdots\\
            \tau_{n-1}(u)
        \end{bmatrix},
\label{mum}
\end{align}
and since $\mathsf{V}_n$ is invertible, \eqref{equ:PolyTuDiagCoeffs8v} follows. 
\end{proof}

For each $i=0,1,\ldots,n-1$, we let
\be
    m^{(i)}_n(x) : = \prod_{\substack{k=0\\[0.05cm] k\neq i}}^{n-1}(x-\mu_k),
\label{minx}
\ee
and note that
\be
 m_n(x)=(x-\mu_i)m^{(i)}_n(x),\qquad m^{(i)}_n(\mu_j)=\delta_{ij}(-1)^i(2n-2)!!\binom{n-1}{i}^{\!-1}.
\label{mimuj}
\ee
Since $m^{(i)}_n(\mu_i)\neq0$ for all $i$, we can renormalise the polynomials (\ref{minx}) as
\be
 \hat{m}^{(i)}_n(x) : =\frac{m^{(i)}_n(x)}{m^{(i)}_n(\mu_i)},\qquad i=0,1,\ldots,n-1.
\ee
In terms of these polynomials, we now define
\be
 \mathsf{p}_i: =\hat{m}^{(i)}_n(h_n),\qquad i=0,1,\ldots,n-1.
\label{pi}
\ee
Using standard arguments, we have the following result.
\begin{prop}
$\{ \mathsf{p}_i\,|\,i=0,1,\ldots,n-1\}$ is a complete set of orthogonal idempotents:
\be
 \sum_{i=0}^{n-1}\mathsf{p}_i=\idn,\qquad \mathsf{p}_j\mathsf{p}_k=\delta_{jk}\mathsf{p}_k,\qquad
  \forall\,j,k\in\{0,1,\ldots,n-1\}.
\ee
\end{prop}
\begin{proof}
By construction, each $\hat{m}^{(i)}_n(x)$ is a polynomial of degree $n-1$, 
so $\sum_{i=0}^{n-1}\hat{m}^{(i)}_n(x)$ is a polynomial of degree at most $n-1$.
Since $\hat{m}^{(i)}_n(\mu_j)=\delta_{ij}$ for all $i,j$, we have $\sum_{i=0}^{n-1}\hat{m}^{(i)}_n(\mu_j)=1$
for every $j\in\{0,1,\ldots,n-1\}$, and since $|\{\mu_0,\mu_1,\ldots,\mu_{n-1}\}|=n$, it follows that 
$\sum_{i=0}^{n-1}\hat{m}^{(i)}_n(x)=1$, hence $\sum_{i=0}^{n-1}\mathsf{p}_i=\idn$.
For $j\neq k$, we have
\be
 \mathsf{p}_j\mathsf{p}_k=\frac{m_n(h_n)}{m^{(j)}_n(\mu_j)\,m^{(k)}_n(\mu_k)}
  \prod_{\substack{i=0\\[0.05cm] i\neq j,k}}^{n-1}(h_n-\mu_i\idn)=0.
\ee
Finally, for each $k=0,1,\ldots,n-1$, we have
\be
 \big(m^{(k)}_n(x)-m^{(k)}_n(\mu_k)\big)\big|_{x=\mu_k}=0,
\ee
so
\be
 m^{(k)}_n(x)-m^{(k)}_n(\mu_k)=(x-\mu_k)q_k(x)
\ee
for some polynomial $q_k(x)$. It follows that
\be
 \mathsf{p}_k\mathsf{p}_k-\mathsf{p}_k
 =\frac{(h_n-\mu_k\idn)q_k(h_n)}{m^{(k)}_n(\mu_k)}\,\hat{m}^{(k)}_n(h_n)
 =\frac{q_k(h_n)}{\big(m^{(k)}_n(\mu_k)\big)^2}\,m_n(h_n)
 =0.
\ee
\end{proof}
\begin{lem}\label{lem:kerpi}
For each $i=0,1,\ldots,n-1$, we have
\be
 \mathrm{ker}(\mathsf{p}_i)=\bigoplus_{\substack{k=0\\[.05cm] k\neq i}}^{n-1}\Vc_{n,k},\qquad
 \mathrm{im}(\mathsf{p}_i)=\Vc_{n,i}.
\ee
\end{lem}
\begin{proof}
Let $i,k\in\{0,1,\ldots,n-1\}$ and $v\in\Vc_{n,k}$. Then,
\be
 \mathsf{p}_iv=\hat{m}^{(i)}_n(\mu_k)v=\delta_{ik}v, 
\ee
and since $\mathsf{p}_i$ is an idempotent, the result follows.
\end{proof}
\begin{prop}
For every $n\in \Nbb$ and all $u\in\Fbb$, we have
\be
 \check{T}_n(u)=\sum_{i=0}^{n-1}\lambda_{i}(u)\mathsf{p}_i
 =\sum_{i=0}^{n-1}\binom{n-1}{i}\frac{(-1)^i\lambda_i(u)}{(2n-2)!!}\,m^{(i)}_n(h_n).
\ee
\end{prop}
\begin{proof}
The result follows from Proposition \ref{prop:Tsl}, Corollary \ref{cor:TVk}, Lemma \ref{lem:kerpi}
and (\ref{mimuj}).
\end{proof}
Alternative expressions for $\check{T}_n(u)$ can be obtained, for example by evaluating the inverse 
$\mathsf{V}_n^{-1}$ explicitly and using
\be
  m_n(x) 
 =x^n-\sum_{i,j=0}^{n-1}[\mathsf{V}_n^{-1}]_{i-1,j-1}\mu_j^nx^i.
\ee
We thus conjecture that $\check{T}_n(u)$ admits the following expression
in terms of the {\em double factorial binomial coefficient} \cite{GQ12},
\be
 \left(\!\!\!\binom{n_1}{n_2}\!\!\!\right): =\frac{n_1!!}{n_2!!\,(n_1-n_2)!!}.
\ee
\begin{conj}\label{conj:Tmu}
For every $n\in \Nbb$ and all $u\in\Fbb$, we have
\be
 \check{T}_n(u)=2\sum_{k=0}^{n-1}(1+u^2)^{n-k}(2u)^k
  \sum_{j=0}^{\lfloor\frac{k}{2}\rfloor}\left(\!\!\!\binom{n-k-2+2j}{2j}\!\!\!\right)\frac{(-1)^jm_{k-2j}(h_n)}{(k-2j)!}.
\label{Tdouble}
\ee
\end{conj}
\noindent
We have verified Conjecture \ref{conj:Tmu} for $n=1,\ldots,180$.

\end{document}